\documentclass[12pt]{article}
\usepackage[margin=1in]{geometry}
\usepackage{amsmath}
\usepackage{amssymb}
\usepackage{graphicx,psfrag,epsf}
\usepackage{enumerate}
\usepackage{booktabs}
\usepackage{breqn}
\usepackage[normalem]{ulem}

\usepackage{titletoc}
\usepackage[page]{appendix}

\usepackage{url} 
\usepackage{float}
\usepackage[table,xcdraw,dvipsnames]{xcolor}

\newcommand\cp{\color{purple}}

\usepackage{cite}  

\usepackage{multirow}
\usepackage{lipsum}
\usepackage{amsfonts}
\usepackage{cite}      
                  
\usepackage{graphicx}   

\usepackage{amsthm}
\RequirePackage[colorlinks,citecolor=blue,urlcolor=blue]{hyperref}
\usepackage{amssymb,calc}
\usepackage{thmtools}
\usepackage{mathrsfs}
\usepackage{mathtools}
\usepackage{bm}
\usepackage{bbm}
\usepackage{caption}
\usepackage{subcaption}

\usepackage{enumerate}
\usepackage{rotating}
\RequirePackage{cleveref}
\usepackage{hyphenat}
\usepackage[numbers]{natbib}
\usepackage{caption,subcaption}
\usepackage{tikz}
\usepackage{pgfplots}

\usepackage[T1]{fontenc}
\usepackage{textcomp} 
\usepackage{lmodern}

\usepackage{listings}
\usepackage{algorithm}
\usepackage[noend]{algpseudocode}
\makeatletter
\renewcommand{\ALG@name}{Algorithm}
\makeatother

\pgfplotsset{width=10cm}

\tikzset{declare function={gamma(\x)=sqrt(2*pi)*\x^(\x-0.5)*exp(-\x)*exp(1/(12*\x));}}
\tikzset{declare function={tpdf(\x,\nu)=gamma(0.5*(\nu+1))/(sqrt(pi*\nu)*gamma(\nu/2))*(1+\x^2/\nu)^(-(\nu+1)/2);}}
\tikzset{declare function={invgampdf(\x,\a,\b)=(\b/\x)^\a/\x/gamma(\a)*exp(-\b/\x);}}

\newcommand{\nhphantom}[1]{\ifmmode\settowidth{\dimen0}{$#1$}\else\settowidth{\dimen0}{#1}\fi\hspace*{-\dimen0}}

\usetikzlibrary{patterns}
\tikzset{
	hatch distance/.store in=\hatchdistance,
	hatch distance=5pt,
	hatch thickness/.store in=\hatchthickness,
	hatch thickness=0.5pt,
}

\newcommand{\wh}[1]{\widehat{#1}}

\definecolor{pink}{rgb}{0.9, 0.17, 0.31}

\newcommand\mT{\mathcal T}

\newcommand\mX{\mathcal X}

\newcommand\mE{\mathcal E}

\renewcommand\P{\mathbb P}

\newcommand\R{\mathbb R}
\renewcommand\b{\bm{\beta}}

\newcommand\bx{\bm x}

\newcommand\bT{\mathbb T}

\newcommand{\wt}[1]{\widetilde{#1}}

\newtheorem{lemma}{Lemma}

\newtheorem{ex}{Example}

\theoremstyle{definition}
\newtheorem{defn}{Definition}[section]

\renewcommand{\nhphantom}[1]{\ifmmode\settowidth{\dimen0}{$#1$}\else\settowidth{\dimen0}{#1}\fi\hspace*{-\dimen0}}

\numberwithin{equation}{section}

\newtheorem{thm}{Theorem}[section]

\newtheorem{remark}{Remark}

\crefname{thm}{Theorem}{Theorems}
\crefname{prop}{Proposition}{Propositions}
\crefname{lem}{Lemma}{Lemmas}
\crefname{coro}{Corollary}{Corollaries}
\crefname{add}{Addendum}{Addendums}
\crefname{asm}{Assumption}{Assumptions}
\crefname{alg}{Algorithm}{Algorithms}
\crefname{proc}{Procedure}{Procedures}
\crefname{exe}{Exercise}{Exercises}
\crefname{exa}{Example}{Examples}
\crefname{prob}{Problem}{Problems}
\crefname{section}{Section}{Sections}
\crefname{subsection}{Section}{Sections}

\def\argmax{\mathop{\arg\max}}

\allowdisplaybreaks[3]

\def\spacingset#1{\renewcommand{\baselinestretch}%
{#1}\small\normalsize} \spacingset{1.4}

\usepackage{graphicx} 

\title{Adaptive Uncertainty Quantification for Generative AI}
\author{Jungeum Kim\footnote{Both authors contributed equally. 
} $^\ddagger$, Sean O'Hagan$^*$\footnote{Department of Statistics, University of Chicago}, and Veronika Ro\v{c}kov\'{a}\footnote{Booth School of Business, University of Chicago}}

\begin{document}

\maketitle
\vspace{-1cm}
\begin{abstract}
This work is concerned with  conformal prediction in contemporary applications (including generative AI) where a black-box model has been trained on data that are not accessible to the user.  Mirroring split-conformal inference, we design a wrapper around a black-box algorithm which  calibrates conformity scores.
This calibration is local and proceeds in two stages by  first adaptively partitioning the predictor space into groups and then calibrating sectionally group by group. 
Adaptive partitioning (self-grouping) is achieved by fitting a robust regression tree to the conformity scores on the calibration set.
This new tree variant is designed in such a way that adding a single new observation does not change the tree fit with overwhelmingly large probability. This add-one-in robustness property allows us to conclude  a finite sample group-conditional coverage guarantee, a refinement of the marginal guarantee. In addition, unlike traditional split-conformal inference, adaptive splitting and within-group calibration yields  adaptive bands which can stretch and shrink locally.
We demonstrate benefits of 
local tightening  on several simulated as well as real examples using non-parametric regression. Finally, we consider two contemporary classification applications for obtaining uncertainty quantification around GPT-4o predictions.  We conformalize skin disease diagnoses based on self-reported symptoms as well as predicted states of U.S. legislators based on summaries of their ideology. We demonstrate substantial local tightening of the uncertainty sets while attaining similar marginal coverage.

\end{abstract}


\clearpage
 
\spacingset{1.7}

\section{Introduction}
{Artificial Intelligence (AI) continues to permeate swiftly through all facets of society, including high-stakes decision domains such as  healthcare \citep{dave2023chatgpt,yeo2023assessing,antaki2023evaluating},  justice \citep{contini2020artificial,sourdin2018judge,scherer2019artificial}, and/or finance \citep{vesna2021challenges,mhlanga2021financial}. The practical utility of AI in decision-making   depends critically on the ability  of  the underlying foundational model to  effectively communicate uncertainty.

Large language models (LLMs)  may be able to convey uncertainty   through repeated prompting \citep{xiong2024can,shrivastava2025language},   self-consistency checks \citep{wang2023selfconsistency},  adversarially perturbed inputs \citep{zhang2023sac}, or by simply  asking LLM to verbalize its own uncertainty  \citep{lin2022teaching,tian2023just,xiong2024can}. 
While these heuristic approaches (see \citep{huang2024survey} for more references)  seek uncertainty quantification through internal reproducibility checks, we focus on  external evaluations   through a  labeled calibration dataset. In particular, this work develops a predictive uncertainty quantification method in the context of contemporary applications involving black-box predictive systems (such as generative models including ChatGPT \citep{brown2020language}, BERT \citep{devlin2018bert}, LLaMA \citep{touvron2023llamaopenefficientfoundation}, and Watson Assistant \citep{ferrucci2010building}) based on conformal prediction.
}
  
  Conformal prediction \citep{vovk2005algorithmic} is arguably one of the most widely used tools for quantification of predictive uncertainty. 
It has been used in the context of  LLMs to provide statistical guarantees for the  factuality of model responses \citep{cherian2024large,mohri2024language,li2025towards}. In our work, we revisit the more traditional statistical frameworks of regression and classification in situations where a  proprietary labeled  dataset is available to calibrate predictions from a model trained on data that are unavailable to the user. 
In the absence of the training dataset, split-conformal prediction (isolating fitting from ranking) is a viable alternative to full conformal prediction  \citep{lei2015conformal, lei2018distribution} as it avoids model refitting on augmented training data. Split-conformal prediction ranks conformity scores on the calibration dataset and guarantees finite-sample marginal coverage without distributional assumptions or any knowledge of
the actual mechanism inside the pre-trained model. While standard conformal prediction bands are marginal over the covariates, and thereby overly conservative in many local regions, a \emph{conditional} coverage guarantee is not achievable in general \citep{lei2014distribution,vovk2012conditional,barber2021limits}. This has motivated the development of various refinements that achieve some relaxed form of conditional coverage \citep{vovk2012conditional, gibbs2023conformal,van2024self} such as group-conformal prediction \citep{lei2013distribution,barber2021limits}. 

Locally adaptive conformal prediction bands have been sought that shrink and stretch  depending on the difficulty of the prediction problem at each local point \citep{lei2018distribution, romano2019conformalized, chernozhukov2021distributional}. 
In particular, \citet{lei2018distribution} locally enhance conformal prediction by rescaling the conformity scores with a function fitted on the  absolute residuals of the training data. This approach is not feasible without access to training data and, moreover, can lead to underestimation of the prediction error when the black-box model has overfitted the training data  \citep{romano2019conformalized}.
Alternatively, \citet{romano2019conformalized} conformalize quantile estimates so that the tightness of the prediction bands depends on the tightness (conformity) of the quantile black-box model in use. \citet{rossellini2023integrating}  further build on this by incorporating the uncertainty in the quantile estimation into the prediction intervals. \citet{chernozhukov2021distributional} generalize this conformalized quantile regression as a fully quantile-rank-based method by using a probability integral transform with a distributional regression estimator (a model of the conditional CDF). These approaches require an auxiliary function (residuals, quantiles, or CDFs), which is obtained using the training dataset. In this work, we concern ourselves with the modern scenario where the users are provided a basic black-box predictive model (without quantiles) and {\em do not} have access to the training data.
Note that this setup aligns with the modern foundational models whose massive training datasets are obscured from the user.


Another way of tightening  prediction intervals is through conditioning on a finite set of groups. \citet{lei2014distribution} apply conformal prediction to each local set for kernel density estimation and study its asymptotic conditional validity. The locality is achieved by an equally spaced partition of the data space, where the data is bisected: one part is used for tuning the spacing size, and the other for model fitting and conformal prediction.  Our approach allows for non-equispaced partitions which are capable of (local) adaptation \citep{rovckova2023ideal, castillo2021uncertainty}. When a good partition is already known, one can consider existing methods such as those described by \citep{gibbs2023conformal, barber2021limits}. \citet{barber2021limits} consider a countable set of potentially overlapping groups with sufficiently large probability mass. The prediction band for a data point is determined by the largest prediction band obtained by applying conformal prediction to each group that contains the data point. \citet{gibbs2023conformal} reformulate conditional coverage as  coverage over covariate shifts. When the probability of the test data is tilted by an indicator function of a subset in the covariate space, this guarantee becomes equivalent to a set conditional coverage guarantee.

Our work is built upon these previous efforts in the sense that we promote local adaptivity through the lens of group-conditioned conformal prediction. {In this work,} we {\em do not} assume that the groups are known in advance. We enhance group-conformal inference by using {\em data-adaptive} grouping driven by a partitioning model fitted to  the conformity scores of the calibration set.
Enabling self-grouping in conditional conformal prediction with continuous predictors is a step forward towards improving local adaptivity in a wider range of settings. Our \emph{self-grouping} is cast as a tree-based supervised learning  where each group is represented by a bin (box) of a box-shaped partition of the covariate space. We then apply the standard split-conformal prediction inside each newfound bin. Equivalently, this translates as constructing the prediction intervals by adding an entire \emph{step function} as opposed to only a constant (as would be the case in split-conformal). 
Motivated by the fact that (dyadic) trees can adapt to varying levels of local smoothness \citep{rovckova2023ideal}, our procedure (called Conformal Tree) deploys dyadic trees for locally adaptive prediction band construction. Since the training set is not available in our scenario, we cannot re-calibrate conformity scores by fitting a (tree or other) model to the conformity scores on the training data (as was done in \citep{lei2018distribution}). 

{


Exchangeability {among the calibration and test data points} is the key assumption that ensures validity of conformal prediction sets without requiring distributional assumptions, by ensuring uniformity of {the ranks of} conformity scores. However, conditioning on groups learned only from the calibration set{, excluding the test data,} may no longer ensure such an exchangeability property. To obviate the need to refit the partition model for any possible realization of testing data,
 we propose a modification of the classical CART algorithm. First, we allow only for splits at dyadic rationals (not at observed values) that are not subject to shift after adding more data. Second, we propose a robust  range-based split criterion and a stopping rule based on a node's size.
 We show that adding a new observation changes the estimated partition only with a small probability controlled by a chosen tuning parameter. This property allows us to conclude  a group-conditional coverage lower bound at a level which only slightly (indistinguishably in practice) deviates from the chosen coverage level. 
The distribution-free aspect distinguishes our work from \citep{wu2025conditional,kiyani2024conformal,guan2023localized}, who developed local conformal prediction methods under suitable distributional assumptions. After the completion of this manuscript, several  related independent approaches were brought to our attention  \citep{martinez2024identifying,izbicki2022cd,cabezas2025regression}. The approach  \citet{martinez2024identifying} is most similar to our work in that it also fits a regression tree on the conformity scores. However, our same-set approach is unique as  it uses the entire calibration dataset for both calibration and partitioning. To achieve grouping, all of these previous works require additional data on top of  the calibration data. The driving force behind our same-set approach is  our new partitioning tree method  that exhibits a unique features that we call "unchangeability", i.e. robustness to adding a new observation. Our same-set approach is particularly appealing when the calibration dataset is not overwhelmingly large 

We account for, and protect against, any systematic differences between calibration and training residuals that might occur due to covariate shifts or model fitting strategies (e.g. intentional overfitting by interpolation). Our predictive interval nearly maintains the desired group-conditional coverage level (see Lemma \ref{lem:conditional_trick} and Theorem \ref{thm:conditional_gu}), paying a negligible price in exchange for the ability to re-use the calibration data. Second, compared to a full-conformal-style prediction band that applies the partition mechanism on the union of the calibration set and the test data points (and thereby would have to be re-fitted for all values of the training conformity score), our robust tree approach allows fitting the partition model only once.
Our computational approach is different from a related approach by  \citet{gibbs2023conformal} who perform conditional quantile estimation over certain function classes, avoiding refitting on augmented data by exploiting monotonicity properties and dual formulation of quantile regression.

We illustrate our framework on several  simulated, as well as real, datasets that are popular in the  literature on conformal prediction for non-parametric regression. For our main illustration, however, we choose two somewhat unique and "nonconformist" applications in the context of generative AI. Our goal is to develop a protective cushion around statements obtained from generative large language models (such as ChatGPT)  for effective communication of uncertainty.
We assume that users' query domain (the domain of the covariates) is   simple  and structured and much smaller in scope compared to the vast domain on which foundation models are trained. The user has access to calibration and test datasets that are specific to the query domain. Instead of laborious fine-tuning of the language model, the user can use this dataset to calibrate conformity scores in order to turn point prediction into set prediction. 

In Section \ref{sec:legislators}, ChatGPT is used to predict the state of a legislator based on a two-dimensional summary of their ideological positions.
In Section \ref{sec:skin_ex} we emulate a scenario where a user might like to seek diagnosis of a skin condition based on self-reported symptom queries. The user who might otherwise rely on a single response is now provided with a set of plausible diagnoses which is calibrated to have the desired coverage. Our empirical results show that conformalizing ChatGPT predictions using our adaptive approach  (1) yields tighter sets (compared to split-conformal) for most test observations, and  (2) is able to communicate "I do not know" through wide sets in cases when the language model is likely to "hallucinate" its query response.  In particular, for the diagnosis example in Section \ref{sec:skin_ex}, 96\% of the test sets have prediction sets that are the same size or smaller than split-conformal. Similarly, for the politics example in Section  \ref{sec:legislators}, 78.4\% of the test prediction sets from Conformal Tree are the same size or smaller than those yielded by split-conformal. Our paper is not meant to encourage users to rely on ChatGPT for predictions. It is only meant to illustrate that if one were to do so, proper accounting for uncertainty is needed.

The paper is structured as follows. In Section \ref{sec:CT}, we introduce our self-grouping algorithm with tree regression and develop a robust regression tree that {is "unchangeable" with high probability after adding an additional   testing data point.} Theoretical justifications for the robust regression tree and Conformal Tree are presented in Section \ref{sec:theory}. Section \ref{sec:exp} demonstrates
 the performance improvement (smaller prediction set) on a number of simulated and standard benchmark datasets. Section \ref{sec:LLM} describes an application of our method to large language models with real datasets. The paper concludes with Section \ref{sec:concluding}.


\paragraph{Notations} 
We use $\P$ to denote probability over the whole dataset $\{(X_i,Y_i)\}_{i=1}^{n+1}$, where $\{(X_i,Y_i)\}_{i=1}^n$ is the calibration dataset. For a black-box model $\hat{\mu}(x)$, define the conformity score of a point $(x,y)$ by $S(x,y) = S(x,y,\hat{\mu}).$ Assuming that $\hat{\mu}$ is trained on the training dataset, independent from the calibration and test sets, we can regard $\hat{\mu}$ and $S$ as fixed functions. Define $S_i = S(X_i,Y_i)$ for $i=1,...,n+1$ {and $\mathcal{D}_{1:n+1} = \{(X_i,S_i)\}_{i=1}^{n+1}$. Denote a subset of $\mathcal{D}_{1:n+1}$ by $\mathcal{D}_I=\{(X_i,S_i)\}_{i\in I}$ for $I\subseteq[n+1]\equiv \{1,2,\dots, n+1\}$.}

\section{Conformal Tree}\label{sec:CT}

The core idea  behind our self-grouping conformal prediction method is to learn the group information by applying a tree model to the conformity scores of the calibration dataset (Section \ref{sec:self_group}). Since this ordinarily may violate the exchangeability assumption between the test and calibration data points, as such a tree would only be fit on the calibration data but not on the test point, we develop a tree algorithm that is robust to adding a new test observation (Section \ref{sec:rob_tree}). 
\subsection{Calibration by Self-grouping}\label{sec:self_group}
Suppose that we are given a trained black-box model $\hat{\mu}(\cdot)$, a conformity score function $S:\mathcal{X}\times \mathcal{Y}\rightarrow \R$ and a calibration dataset $\{(X_i,Y_i)\}_{i=1}^n$. For a new test point $X_{n+1}$, we want to predict the response $Y_{n+1}$ under the assumption  that $(X_1,Y_1),....,(X_{n+1},Y_{n+1})$ are {i.i.d.} 
The split-conformal {prediction interval (or a prediction set)} for $1-\alpha$ coverage is defined as 
\[\wh{C}_n(X_{n+1}) = \{y: S(X_{n+1}, y)\leq S^*(\mathcal{D}_{1:n})\},\]where $S^*(\mathcal{D})$ is the $(\lceil (n+1)\cdot (1-\alpha)\rceil/n)$-quantile of the conformity scores {in} $\mathcal{D}$ and where $\mathcal{D}_{1:n} = \{(X_i,S_i)\}_{i=1}^{n}$. For $\wh{C}_n (X_{n+1})$, the coverage is guaranteed \citep{lei2014distribution, vovk2005algorithmic} to satisfy
\[
\mathbb{P}\{Y_{n+1}\in\wh{C}_n(X_{n+1})\}\geq 1-\alpha.
\]
Note that the prediction interval construction relies on the model's performance in the overall input space. For example, with the absolute residual score $S(x,y)=|y-\hat{\mu}(x)|$, the set $\wh{C}_n (X_{n+1})$ is equivalent to $[\hat{\mu}(X_{n+1})-\hat{q}, \hat{\mu}(X_{n+1})+\hat{q}]$, where a constant $\hat{q}=S^*(\mathcal{D}_{1:n})$ is added and subtracted over the data domain $\mathcal{X}$. 

To locally sharpen the prediction interval, we can consider adding and subtracting a constant that differs depending on the location of $X_{n+1}$. Assume that the data space $\mathcal{X} = \cup_{k=1}^K \mathscr{X}_k$ is partitioned into non-overlapping groups $\mathfrak{X}=\{\mathscr{X}_1,\ldots,\mathscr{X}_K\}$. We allow for categorical covariates which  naturally serve as a group indicator. If $X_{n+1}\in \mathscr{X}_k$, we can calibrate the conformal interval based on only data points inside $ \mathscr{X}_k$ to obtain a group-conditional coverage guarantee. In other words, defining
\begin{equation}\label{eq:conditional_band}
    \wh{C}^{\mathfrak{X}}_{n} (X_{n+1}) = \{y: S(X_{n+1}, y)\leq S^*(\mathcal{D}_{1:n}^{k,\mathfrak{X}})\},
\end{equation}
where $\mathcal{D}_{1:n}^{k,\mathfrak{X}}: = \{(X_i,{S}_i)\in \mathcal{D}_{1:n}\mid X_i\in \mathscr{X}_k\}$, we have the group-conditional coverage guarantee 
\[\mathbb{P}\{Y_{n+1}\in\wh{C}^{\mathfrak{X}}_{n}(X_{n+1})\mid X_{n+1}\in \mathscr{X}_k\}\geq 1-\alpha\] for every $k\in \{1,...,K\}$. This is because the exchangeability among $(X_{n+1},S(X_{n+1},Y_{n+1}))$ and other $(X_i,S(X_i,Y_i))$ is still preserved, conditioning on the event that they are in $\mX_k$ {(See, Lemma \ref{lem:foundation}}). With the absolute residual score, $\wh{C}^{\mathfrak{X}}_{n} (X_{n+1})$ is equivalent to $[\hat{\mu}(X_{n+1})-\hat{q}_k, \hat{\mu}(X_{n+1})+\hat{q}_k]$, where $\hat{q}_k=S^*(\mathcal{D}^{k,\mathfrak{X}}_{1:n})$ is defined for each group $\mathscr{X}_k$. 

Our method, called Conformal Tree, is described inside Algorithm \ref{alg:conf_tree}. This procedure  has an additional partition step compared to group-conformal inference \citep{gibbs2023conformal, barber2021limits} because we {\em do not} assume that the grouping information is known \emph{a priori}. We apply a regression tree on the conformity scores of the calibration dataset $\mathcal{D}_{1:n}$ using 
a potentially high-dimensional vector of both continuous and  discrete covariates. Denote by $\wh \mT(\mathcal D_{1:n})\in \bT_{K_{\rm max}}$ the tree fitted on $\mathcal{D}_{1:n}$, where $\bT_{K_{\rm max}}$ is the class of binary dyadic trees with {maximum} ${K_{\rm max}}$ leaf nodes (defined formally later). The partition induced by the tree $\wh \mT(\mathcal D_{1:n})$ is denoted by $\wh{\mathfrak{X}}(\mathcal{D}_{1:n})$, where calibration is performed locally through \eqref{eq:conditional_band} with $\mathfrak{X}=\wh{\mathfrak{X}}(\mathcal{D}_{1:n})$. 
{The fact that grouping information depends only on the calibration data, and not the test point, can potentially lead to a violation of group-conditional exchangeability for the conformity scores  (see, Lemma \ref{lem:exch}).} However, if we were to obtain a partition $\wh{\mathfrak{X}}(\mathcal{D}_{1:n+1})$ based on the union of calibration and testing data $\mathcal{D}_{1:n+1}$, the triples 
$\left(X_1,Y_1, S(X_1,Y_1)\right),...,\left(X_{n+1},Y_{n+1},S(X_{n+1},Y_{n+1})\right)$ would still be exchangeable when conditioned on $\wh{\mathfrak{X}}(\mathcal{D}_{1:n+1})$. {Motivated by this observation,} we propose a grouping mechanism based on a robust tree approach (Section \ref{sec:rob_tree}), which guarantees that the partition does not change too often after adding one more observation, i.e. $\wh{\mathfrak{X}}(\mathcal{D}_{1:n})=\wh{\mathfrak{X}}(\mathcal{D}_{1:n+1})$ with high probability. {Such "unchangeability" property guarantees that the conditional coverage  (that would be guaranteed if  $\wh{\mathfrak{X}}(\mathcal{D}_{1:n+1})$   were used for splitting) is inherited with high probability  when $\wh{\mathfrak{X}}(\mathcal{D}_{1:n})$ is used instead.}

\begin{algorithm}[!t]
\caption{Conformal Tree}\label{alg:conf_tree}
\spacingset{1.2}
\begin{algorithmic}[1]
\Require Fitted model $\hat{\mu}$, calibration dataset $\mathcal{D}_{1:n}$, tree hyperparameters ($m,{K_{\rm max}},\eta=0.05$)
\State Fit a tree $\hat \mT(\mathcal D_{1:n})$ with Algorithm \ref{alg:robust_tree} on ${\mathcal{D}_{1:n} =}\{(X_i,S(X_i,Y_i)):i\in[n]\}.$ 
\State Define a partition $\wh{\mathfrak{X}}(\mathcal{D}_{1:n})=\{\mathscr{X}_1,\ldots,\mathscr{X}_K\}$ induced by leaves of $\hat \mT(\mathcal D_{1:n}).$\label{ln:rule1}
\State For $k=1,...,K$, compute $S^*(\mathcal{D}_{1:n}^{k,\wh{\mathfrak{X}}(\mathcal{D}_{1:n})})$, the $\lceil (1-\alpha)\cdot (m_k-2)+1\rceil/m_k$-th smallest value of $\{S(X_i,Y_i): i \in[n],X_i\in\mathscr{X}_k\}$, where $m_k$ is the cardinality of this set.\label{ln:rule2}
\State Construct a prediction set $\wh{C}_{n,m}$ as 
\[
\wh{C}_{n,m}(x) = \{y:S(X_{n+1},y)\leq S^*(x)\},
\]
where $S^*(x)=S^*(\mathcal{D}_{1:n}^{k,\wh{\mathfrak{X}}(\mathcal{D}_{1:n})}):x\in\mathscr{X}_{k}$\,.
\end{algorithmic}
\end{algorithm}

\begin{remark}
    { We highlight that Conformal Tree also identifies relevant subsets of the covariate space in which conditional coverage is guaranteed. The tree-based nature of group creation endorses identifiable and explainable groups which can often be understood in terms of some physical meaning. For example, one could identify that a model has more inherent uncertainty (and thus larger prediction sets) when classifying members of a specific race and gender combination. Refer to Figure~\ref{fig:derm-tree-structure2} where this is demonstrated on a skin disease classification example.}
\end{remark}

\subsection{Robust Regression Trees}\label{sec:rob_tree}

Regression trees \citep{breiman84classification} have been widely used in the nonparametric regression literature due to their interpretability. Trees employ a hierarchical set of rules in order to partition the data into subsets on which the response variable is relatively constant. Typical approaches to fitting trees involve greedily descending from the root note and making splitting decisions in order to maximize a utility criterion, subject to regularization constraints.

For many popular criteria used to assess tree fitting, such as variance reduction, it is hard to precisely understand the effect of including a single new data point on the resulting tree structure. We introduce a novel criterion that aims to produce a tree that is robust to the inclusion of an additional {i.i.d.} data point with high probability.\footnote{The result is also established with an additional exchangeable data point.} 


{We focus on dyadic trees where splits along covariates may only be placed at specific predetermined locations (dyadic rationals for continuous predictors rescaled to a unit interval). In general, our implementation allows for discrete predictors, provided that users define a way of splitting them. For example, we can map discrete covariates to a subset of the real line and split this component of the domain as if they were continuous.}

We define a binary tree $\mT$ as a collection of hierarchically organized nodes $(l,k)$ such that $(l,k)\in \mT\Rightarrow (j,\lceil k/2^{l-j}\rceil)\in \mT$ for $j=0,...,l-1$. Here $l$ denotes the depth, $k$ denotes the horizontal position, and the root node is defined to be $(0,0)$. The left and right children of a node $(l,k)$ are given by $(l+1,2k)$ and $(l+1,2k+1)$, respectively.  A node $(l,k)\in \mT$ is called an external or leaf node if it has no children. We denote the set of external nodes of $\mT$ by $\mT_{\rm leaves}$ and the internal nodes $\mT_{\rm int} = \mT\setminus \mT_{\rm leaves}$. The internal nodes are assigned a split direction while the leaf nodes are not. We denote these split directions by $\mathcal H(\mT)=\{h_{lk}\in\{1,\dots, d\}:(l,k)\in\mT_{int}\}$. We consider only dyadic trees, for which the tree nodes correspond to axis-aligned dyadic splits. We define $\bT_{{K_{\rm max}}}$ as a space of all dyadic binary trees that split at most ${K_{\rm max}}-1$ times in one of the $d$ directions yielding up to ${K_{\rm max}}$ leaves. The partition $\frak{X}(\mT)$, induced by both $\mT$ and $\mathcal H(\mT)$,  consists of regions associated with the leaf nodes $\mT_{\rm leaves}$. In particular, each leaf node $(l,k)\in \mT_{\rm leaves}$ corresponds to a rectangular region of the input space $\mathcal{X}$, which we denote as $\mathcal{X}_{(l,k)}(\mT, \mathcal{H}(\mT))$. Given data $X_1,\ldots,X_n\in\mathcal{X}$ and scores $S_1,\ldots,S_n\in\mathbb{R}$, we denote the subset of data points that fall inside the  $(l,k)^{th}$ region as $X_{(l,k)}(\mT,\mathcal{H}(\mT))=\{X_i:X_i\in\mathcal{X}_{(l,k)}(\mT, \mathcal{H}(\mT))\}$ and, similarly, $S_{(l,k)}(\mT,\mathcal{H}(\mT))=\{S_i:X_i\in\mathcal{X}_{(l,k)}(\mT, \mathcal{H}(\mT))\}$.

  Our robust tree-fitting criterion is based on a splitting rule that targets the response range (as opposed to deviance). At every split opportunity, we split at the terminal node in the direction that yields the largest range decrease. Below, we formally define   the \emph{range reduction} of a node $(l,k)$ if it were split   dyadically  along the dimension $j$ denoted by $I_R^j(l,k)$. 
\begin{defn}\label{def:basic}
For a tree $\mT\in\bT_{K_{\rm max}}$ with split directions $\mathcal{H}(\mT)$, the range of a node $(l,k)\in\mT_{\rm leaves}$ is given by $R_{lk}(\mT)=\max S_{(l,k)}(\mT,\mathcal{H}(\mT)) - \min S_{(l,k)}(\mT,\mathcal{H}(\mT))\,.$
When a node $(l,k)$ splits along the $j^{\textrm{th}}$ direction, we define the range of the left child as $
R^{j,\rm left}_{lk}(\mT)=\max S_{(l+1,2k)}(\wt\mT,\mathcal{H}^j) - \min S_{(l+1,2k)}(\wt{\mT},\mathcal{H}^j)\,,$
where $\wt{\mT}=\mT\cup\{(l+1,2k),(l+1,2k+1)\}$ and $\mathcal{H}^j =\mathcal{H}(\mT) \cup \{h_{lk}=j\}.$ The range of the right child $R^{j,\rm right}_{lk}(\mT)$ is defined similarly as $S_{(l+1,2k+1)}(\wt\mT,\mathcal{H}^j)$. Then, the range reduction obtained by splitting a node $(l,k)$ in dimension $j$ is given by
\begin{equation}
I_{lk}^j(\mT)=R_{lk}(\mT) - \left[ R^{j,\rm left}_{lk} (\mT)+  R^{j,\rm right}_{lk}(\mT)\right]\,.
\label{eq:robust_criterion}
\end{equation}
\end{defn}

As is typical with CART implementations \citep{ripley2023tree}, we regularize the tree and define a stopping criterion by the following three hyperparameters: the minimal number of samples per leaf $m$, maximum number of leaves ${K_{\rm max}}$, and minimum range reduction rate $\eta=0.05$. A node is called a \emph{candidate node} if there exists a direction along which a split  would not cause the tree to violate any of these criteria.

  \begin{defn}[Candidate node]\label{def:cand} We use the notation from Definition \ref{def:basic}. For a tree $\mT\in\bT_{K_{\rm max}}$ such that $|\mT_{\rm leaves}|<{K_{\rm max}}$, the set of eligible split directions for a node $(l,k)\in \mT_{\rm leaves}$ is defined as 
\begin{align}
    J_{lk}(\mT):=&\left\{j\in[d] :|X_{(l+1,2k)}(\wt\mT,\mathcal{H}^j)|\geq m \text{ and }|X_{(l+1,2k+1 )}(\wt\mT,\mathcal{H}^j)|\geq m \right\}\nonumber\\
    &\cap\left\{j\in[d] :I^j_{lk}(\mT)/R^j_{lk}(\mT) \geq 0.05 \right\},\label{eq:candidacy}
\end{align}
The leaf node $(l,k)$ is called a \emph{candidate node} in $\mT$ if $J_{lk}(\mT)\neq\emptyset$. We denote by $\mathfrak{N}(\mT)$ the set of candidate nodes in a tree $\mT.$

  \end{defn}

In order to grow the tree, we search over all candidate nodes, and choose to split at the candidate node that most reduces the range of response values.

\begin{defn}(Range Reduction)
  The range reduction of a candidate node $(l,k)$ in tree $\mT$ with history $\mathcal{H}(\mT)$ is given by the maximum range reduction over eligible directions 
\begin{equation}
  I_{lk}(\mT)=\max_{j\in J_{lk}(\mT)}I_{lk}^{j}(\mT)
  \label{eq:rangereduction}
  \end{equation}
\end{defn}

We continue to sequentially add splits to the tree by maximizing the range reduction among candidate nodes until either we reach the maximum number of leaves ${K_{\rm max}}$, or there are no candidates remaining, at which time we stop the tree fitting procedure.

\begin{algorithm}[!t]
\caption{Robust Dyadic Regression Tree}\label{alg:robust_tree}
\spacingset{1.2}
\begin{algorithmic}[1]
\Require Dataset $\mathcal{D}=\{(X_i,S_i):i\in[n]\}$, minimum samples per leaf $m$, max leaves ${K_{\rm max}}$
\State Set an initial tree $\mathcal{T} \gets \{(0,0)\}$ \Comment{Initialize the tree as a  root node}
\While{$|\mathcal{T}_{\mathrm{leaves}}|< {K_{\rm max}}$ }
  \State Set $\mathfrak{N}(\mT)$ as the set of all candidate leaves $(l,k)$ as in Definition \ref{def:cand}
  \If {$\mathfrak{N}(\mT)\neq \emptyset$}
  \State Compute $I_{lk}(\mT)$ for all $(l,k)\in \mathfrak{N}(\mT)$ as in \eqref{eq:rangereduction} 
        \State Choose  $(l^*,k^*) \in \argmax_{ (l,k)\in\mathfrak{N}(\mT)} I_{lk}(\mT)$, {denoting the maximizer of \eqref{eq:rangereduction} as $j^*$} \Comment{Breaking ties arbitrarily}
        \State Add children $\mT = \mT \cup \{(l^*+1,2 k^*), (l^*+1,2 k^*+1)\}$
        \State {Add new split direction $\mathcal{H}(\mathcal{T}) = \mathcal{H}(\mathcal{T})\cup \{h_{l^*k^*}=j^*\}$}
\EndIf
\EndWhile
\State \textbf{Return} the final tree $\mT$ 
\end{algorithmic}
\end{algorithm}

Fitting a dyadic robust regression tree using the stopping criterion \eqref{eq:candidacy} and  \eqref{eq:rangereduction} is summarized in Algorithm \ref{alg:robust_tree}. Once we obtain a partition using Algorithm \ref{alg:robust_tree} on the calibration data, we employ conformal prediction by constructing prediction sets with varying sizes for each leaf node as in Algorithm \ref{alg:conf_tree}. In this manner, we have partitioned the data into regions conditionally on which we achieve coverage (Lemma \ref{lem:conditional_trick}). These groups are self-selected in a data-dependent way  to greedily maximize the reduction of the total spread of the grouped conformity scores. This robust tree has an important stability property when used as our partitioning tool for computing the conditional conformity threshold: the partition remains unchanged after an addition of another observation provided that it is not too extreme, and does not modify the candidacy status of its leaf node. Theoretical guarantees for this unchangeability of the partitions and the corresponding conformal coverage guarantee are presented in Section \ref{sec:theory}.

{The new robust tree algorithm can be extended to incorporate the $X_{n+1}$ value (not $y_{n+1}$) for tree fitting. While this variation gives tighter theoretical bounds, it introduces non-negligible additional computational cost that scales linearly with the number of test points. For this reason, we defer it to Section \ref{sec:tighter} in the Supplement.}

\begin{remark}
Our robust tree algorithm resembles greedy CART \citep{breiman84classification}. {Similarly, we require a minimum improvement (range reduction rate) and a minimum number of observations per leaf. Rather than using cross-validation to prune back, we rely on early stopping once we reach ${K_{\rm max}}$ leaf nodes or until splitting would yield less than $m$ observations per leaf node, or an insufficient improvement.} The tuning parameter ${K_{\rm max}}$ is pre-determined, unlike Bayesian implementations of CART \citep{chipman1998bayesian} which are both locally and globally adaptive. {We use a fixed value of 0.05 for the minimum improvement threshold for all of the experiments in this paper.}
\end{remark}
\begin{remark}(Conformal Forests)
\label{rm:forest}\citet{gasparin2024merging} have recently proposed a method for combining conformal sets via a weighted majority vote procedure  which preserves a weakened coverage guarantee. We briefly describe how this can be used to extend Conformal Tree to the random forest setting. One could grow $h$ different trees, for example, by  subsampling the calibration data, randomizing  splitting dimensions, or both. The prediction set corresponding to each tree can be combined using the weighted majority vote procedure \citep{gasparin2024merging}, yielding an aggregated interval. When applied to our robust dyadic trees in Algorithm \ref{alg:robust_tree}, the aggregated interval will have marginal coverage with probability at least $1-2\alpha-2\delta$. Such an aggregated interval may sacrifice conditional coverage in each leaf node, as well as some of the coverage guarantee, in exchange for a smoother and more locally adaptive conformal set. {As dyadic trees split always at the center, we could apply (random) rotation of the covariates when fitting multiple trees similar to \cite{wang2024generative} to diversify the tree fits. Such an approach will require a device of integrating the local weighting mechanism in \cite{wang2024generative} with the the weighted majority vote procedure \citep{gasparin2024merging}. In our experiments (Table \ref{table:tuned-results2}), we show simply the random forest version of conformal tree, using the regular CART with the majority vote in \citep{gasparin2024merging}.} 

\end{remark}

\subsubsection{Full-Conformal with Robust Trees}
Having an algorithm that is robust to adding one observation invites the possibility of performing actual full conformal prediction.
Indeed, when a robust tree serves as a regression model $\hat{\mu}$, not only as a tool for partitioning,  standard full conformal prediction can be obtained efficiently. Consider an augmented dataset $\{(X_i,Y_i)\}_{i=1}^n\cup\{(x,y) \},$ where $(x,y)$ is a possible realization of $(X_{n+1},Y_{n+1})$. Denote by $\hat{\mu}^{(x,y)}$ a tree model fitted on {the augmented dataset.} The conformity score of $(X_i,Y_i)$ is defined by $S_{(x,y)}(X_i,Y_i) = |Y_i-\hat{\mu}^{(x,y)}(X_i)|$. The standard full conformal prediction band for $1-\alpha$ coverage \citep{lei2018distribution} is defined by 
\begin{equation}\label{eq:full_conf}
  \wh{C}_{\rm full} (X_{n+1}) = \{y: S_{(X_{n+1},y)}(X_{n+1}, y)\leq  S^*(\mathcal{D}^{(X_{n+1},y)}_{1:n})\}, 
\end{equation}
where $\mathcal{D}^{(x,y)}_{1:n} = \{(X_i, S_{(x,y)}(X_i,Y_i))\}_{i=1}^n$. Here, for a given $X_{n+1}$, we need to fit the model for all possible $y$ values to compute $\wh{C}_{\rm full} (X_{n+1})$ in \eqref{eq:full_conf}. 

This computational burden is reduced if the fitted model is robust to varying the values of $(x,y)$. Note that our robust tree regression guarantees the partition is the same with high probability (Lemma \ref{lem:tree_same}, with probability in \eqref{eq:tree_same_lemma_bound}). By defining the jump coefficient on each box as the center of the range of the response variables whose covariates fall into the box, we can guarantee with the same high probability that $\hat{\mu}^{(x,y)} = \hat{\mu}_{1:n}$ for any possible $(x,y)$, where $\hat{\mu}_{1:n}$ is the robust tree fitted on  $\{(X_i,Y_i)\}_{i=1}^n$. For this $\hat{\mu}_{1:n}$, define the conformity score $S_n(X_i,Y_i) := |Y_i-\hat{\mu}_{1:n}(X_i)|$. Since the score function does not depend on $(x,y)$, the resulting full conformal prediction band now resembles the split-conformal inference. Indeed, the full conformal prediction band in \eqref{eq:full_conf} becomes equivalent to 
\begin{equation}\label{eq:robust_full1}
  \wh{C}_{\rm full} (X_{n+1}) = \{y: S(X_{n+1}, y)\leq S^*(\mathcal{D}_{1:n})\}, 
\end{equation}
 or equivalently, 
\begin{equation}\label{eq:robust_full2}
  \wh{C}_{\rm full} (X_{n+1}) = [\hat{\mu}_{1:n} -S^*(\mathcal{D}_{1:n}), \hat{\mu}_{1:n} +S^*(\mathcal{D}_{1:n})],
\end{equation}
{where $\mathcal{D}_{1:n}= \{(X_i, S_n(X_i,Y_i))\}_{i=1}^n$ (only in this section).}
Note that $S^*(\mathcal{D}_{1:n})$ is a constant and does not depend on $(x,y)$ but only on $\mathcal{D}_{1:n}$ and $\hat{\mu}_{1:n}$. 

There have been clever attempts to avoid the computation of full conformal prediction bands. When the conformity score is the Bayesian posterior predictive likelihood, \citet{fong2021conformal} bypass the full conformal prediction computation by first sampling from the posterior conditionally on $X_{1:n+1}$, and then importance weighting based on $X_{n+1}.$ \citet{lei2013distribution} simplify the analysis on full conformal prediction bands for kernel density estimation  by sandwiching them with level sets that involve only the training set (not the test set). Our work avoids the computational burden of full-conformal-style prediction bands by developing a robust tree regression algorithm against extreme conformity scores and by applying it only to the calibration set. {\citet{stutz2021learning} modified the loss function when training deep neural networks to consider the resulting effect on the width of the conformal interval for later prediction tasks.} 
While we consider post-training black-box inference, our tree partitioning idea may be applied to sharpen the prediction bands if used during their training phase. 

\section{Theoretical Underpinnings}\label{sec:theory}
Recall that our algorithm (Algorithm \ref{alg:conf_tree}) applies the partitioning method (Algorithm \ref{alg:robust_tree})  on the calibration dataset $\mathcal{D}_{1:n}$, not the full dataset  $\mathcal{D}_{1:n+1}$. This may create some issues with exchangeability when conditioned on the resulting groups.
Fortunately, conditionally on the partition $\wh{\mathfrak{X}}(\mathcal{D}_{1:n})$ estimated by our robust tree approach, 
{the conformity score of the test point is controlled at the desired level with high probability.} 
This property is implied by the "unchangeability" of the estimated partition after adding a test observation, as described below.

\begin{lemma}\label{lem:tree_same}
    Let $(X_1,Y_1),\ldots,(X_{n+1},Y_{n+1})$ be {i.i.d. samples from a joint probability measure $\mathbb{P}$}. Let $\wh{\mathfrak{X}}(\mathcal{D}_{1:j})$ denote the partition resulting from Algorithm \ref{alg:robust_tree} on data $ (X_1,S(X_1,Y_1)),\ldots,$ $(X_{j},S(X_{j},Y_{j})))$ for $j\in\{n,n+1\}$. Denote by $m$ the hyperparameter limiting the minimum number of samples in each leaf node and {$K_{\rm max}<(n+1)/m$ be the maximum number of leaf nodes}. Define
    \begin{equation}
    \delta(n,m) := \frac{2}{m} + \exp\left\{-\left(\frac{n+1}{K_{\rm max}}-m\right)\right\}
    \label{eq:delta_defn}
    \end{equation}
    Then 
    \begin{equation}
    \label{eq:tree_same_lemma_bound}
    \mathbb{P}\{\wh{\mathfrak{X}}(\mathcal{D}_{1:n})=\wh{\mathfrak{X}}(\mathcal{D}_{1:n+1})\} \geq 1-\delta(n,m)\,.
    \end{equation}    
    {If we fix $m=c\cdot (n+1)$ for $c\in(0,1)$ such that $c\cdot(n+1)\in\mathbb{N}$, we have a more interpretable bound $\delta(n,m)\leq (2/m)+e/(2\pi \sqrt{c(1-c)(n+1)})$.}
\end{lemma}
\proof See, Section \ref{pf:tree_same}.

Due to Lemma \ref{lem:tree_same}, it is highly probable that the robust tree partitioning procedure applied on the fuller dataset $\mathcal{D}_{1:n
+1}$ will result in the same partition.

{The following theorem says that we can nearly enjoy the prescribed level of conditional coverage.} Intuitively, this is because the self-grouping procedure is sufficiently robust to the addition of the $(n+1)\mathrm{st}$ data point as a result of Lemma \ref{lem:tree_same}. 

\begin{thm} \label{thm:conditional_gu}
    {Let $(X_1,Y_1),\ldots,(X_{n+1},Y_{n+1})$ be {i.i.d. samples from a joint probability measure $\mathbb{P}$}. Denote by  $\wh{C}_{n,m}(X_{n+1})$ the prediction interval on $X_{n+1}$ constructed by Algorithm \ref{alg:conf_tree}, with a minimum of $m$ samples in each leaf. Denote the obtained partition by $\wh{\mathfrak{X}}(\mathcal{D}_{1:n}) = \{\mathscr{X}_k\}_{k=1}^K.$} Then, we have for any $k\in\{1,\dots, K\},$
    \[
    \mathbb{P}\{Y_{n+1}\in \wh{C}_{n,m} (X_{n+1}) \mid X_{n+1} \in \mathscr{X}_k ,\wh{\mathfrak{X}}(\mathcal{D}_{1:n})\} \geq  1-\alpha-\delta(n,m)\,
    \]where $\delta(n,m)$ is as defined in \eqref{eq:delta_defn}. {Furthermore, provided that the conformity scores have a continuous joint distribution, 
\begin{equation}\label{eq:thm-upper-bd}
\mathbb{P}\{Y_{n+1}\in\wh{C}_{n,m}(X_{n+1})\mid X_{n+1} \in \mathscr{X}_k ,\wh{\mathfrak{X}}(\mathcal{D}_{1:n})\}\leq 1-\alpha+\frac{1}{1+m}+\delta(n,m)\,.
\end{equation}}
\end{thm}
 \proof See Section \ref{sec:proof_conditional_gu}. 

This result assures that our predictive interval nearly maintains the desired group-conditional coverage level, incurring only a minimal cost for the advantage of reusing the calibration data.    Going further, the following theorem states that the  conditional coverage result extends to a marginal coverage result.
\begin{thm}\label{thm:main} 
Let $(X_1,Y_1),\ldots,(X_{n+1},Y_{n+1})$ be {i.i.d. samples from a joint probability measure $\mathbb{P}$}. Denote by  $\wh{C}_{n,m}(X_{n+1})$ the prediction interval on $X_{n+1}$ constructed by Algorithm \ref{alg:conf_tree}, with a minimum of $m$ samples in each leaf. Then, we have 
\begin{equation}
\label{eq:thm-lower-bd}
\mathbb{P}\{Y_{n+1}\in\wh{C}_{n,m}(X_{n+1})\}\geq 1-\alpha-\delta(n,m)\,,
\end{equation}
where $\delta(n,m)$ is as defined in \eqref{eq:delta_defn}.
Furthermore, provided that the conformity scores have a continuous joint distribution, 
\begin{equation}\label{eq:thm-upper-bd}
\mathbb{P}\{Y_{n+1}\in\wh{C}_{n,m}(X_{n+1})\}\leq 1-\alpha+\frac{1}{1+m}+\delta(n,m)\,.
\end{equation}
\end{thm}
\proof See Section \ref{pr:main_thm}. 

 While the bounds in Theorem \ref{thm:main} are not entirely optimistic   
  as we upper-bound the probability of miscoverage by 1 when $\wh{\mathfrak{X}}(\mathcal{D}_{1:n})\neq \wh{\mathfrak{X}}(\mathcal{D}_{1:n+1})$.
We found empirically, however, that the Conformal Tree  yields the desired coverage on most examples we tried (as summarized in Table \ref{table:all-results}). One may also restore the originally intended coverage level $1-\alpha$ by using the $1-\alpha +\delta(n,m)$ quantile (if below one) to define the conformal interval within each leaf region.

\begin{remark}\label{rm:new_algo}
    {Our algorithm   necessitates the robust tree to be fitted only once in order to yield a prediction set for any value of $X_{n+1}$. The sacrifice $\delta(n,m)$ in the coverage  can be improved to just $2/m$ if we refit the tree for each $X_{n+1}$ value of interest. This is because instantiating the partition with  each $X_{n+1}$ prevents the possibility of  changing  the candidacy status of a node. See Section \ref{sec:tighter} for a simple modification to the algorithm and corresponding theoretical results.}
\end{remark}

\section{Simulation and Benchmark Studies}\label{sec:exp}

In this section, we consider regression tasks on simulated and real benchmark data. 
We use the absolute residuals as the conformity score. As our primary evaluation metric, we compare the average length of the prediction interval and proportion of test intervals smaller than that of split conformal (proportion better, or in short, P.B.). We also measure empirical coverage rates on the test dataset. In Section \ref{sec:hyp_sense}, a sensitivity analysis for the hyperparameters $K_{\rm max}$ and $m$ is provided. Additional information on data sets, predictive models, and experiments can be found in Supplement \ref{sec:exp-info}.

\subsection{Conformal Tree is Adaptive}\label{sec:exp_adapt}
{We illustrate Conformal Tree on simulated data examples, modified from previous studies \citep{rossellini2023integrating,rovckova2023ideal}. We let $X\sim \mathcal{U}(0,1)$, and consider the following two data-generating processes with $n=500$:
\begin{itemize}
\item[](Data 1) $Y|X \sim N(3{\rm sin}(4/X+0.2)+1.5,X^2)$ and 
\item[](Data 2) $Y|X \sim N({\rm sin}(X^{-3}),0.1^2)$.
\end{itemize} 
These two setups represent heteroskedasticity for distinct reasons. Data 1 has the noise standard deviation scaling with $X$, causing the distribution of residuals to vary across the domain. Therefore, any single regression (mean) model cannot avoid increasing errors for the larger covariate due to the increasing variation of the data. Data 2 is difficult to model near $X=0$ due to the $X^{-3}$ term inside the $\sin$, causing rapid change in the mean function that requires a lot of data to model accurately. Therefore, due to the lack of model fit, the error is large for the smaller covariate values.

\begin{figure}[!t]
    \centering
    \includegraphics[width=0.88\linewidth]{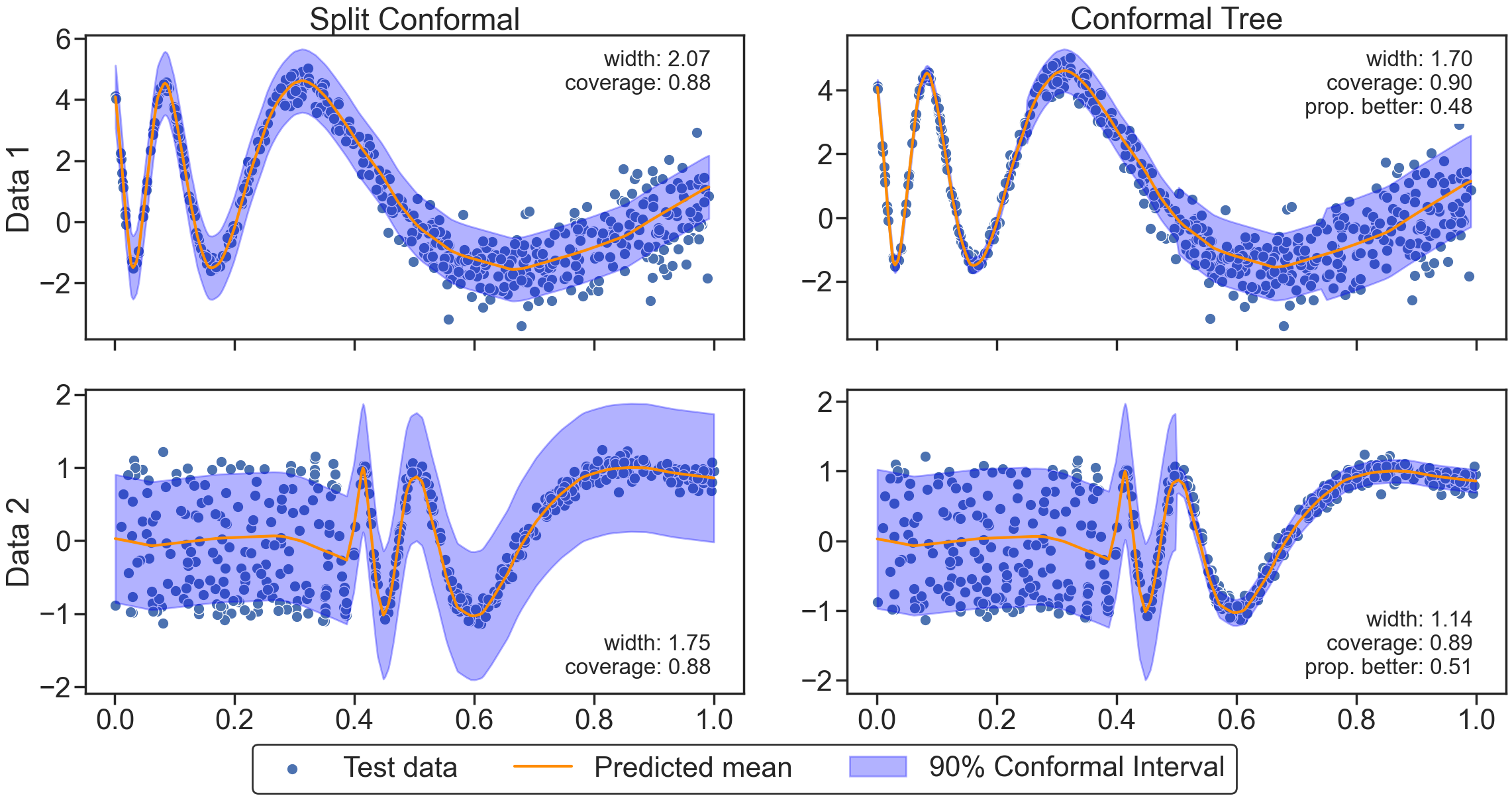}\caption{Conformal Tree can adaptively {reflect locality without} an auxiliary model of quantiles or conformity scores fit on the training dataset. We display the average width, empirical coverage on held-out test data, and proportion better.  The interval corresponds to a single random trial among the 10 random trials.}
    \label{fig:true_data}
\end{figure}

  {We train a regression model $\hat{\mu}$ a fully-connected deep feed forward neural network with a smooth regularization penalty. The trained model behaves as the orange line in Figure \ref{fig:true_data}.} {We synthetically generated datasets of size 1000 and randomly split the data into equally sized calibration and test datasets ten times, and applied each method for each split}. We apply our Conformal Tree method ({$m=100$ and maximum leaves as $5$}) and compare it with the standard split-conformal prediction. Figure \ref{fig:true_data} highlights that Conformal Tree can adaptively reflect local changes without the necessity for an auxiliary model of quantiles or conformity scores on the training dataset. As the partitioning method is robust to the addition of a new iid data point, the obtained prediction intervals with $\alpha=0.1$ still have good coverage at 0.899 for Data 1 and 0.904 for Data 2 when averaged over 10 random trials. For the standard conformal prediction, the average coverage was 0.904 and 0.874, respectively. Furthermore, the average percentage of points for which the interval was narrower is 0.54 for Data 1 and 0.5 for Data 2. 

\subsection{Benchmark Datasets} \label{sec:benchmark}
We consider various real data examples that had previously been used in the literature
for comparing the efficacy of regression prediction intervals \citep{rossellini2023integrating, sesia2020comparison, romano2019conformalized}, {namely, \texttt{bike}, \texttt{bio}, \texttt{community}, \texttt{concrete}, \texttt{homes}, and \texttt{star}. The number of observations in these datasets vary from 200 to about 50,000, and the number of predictor variables ranges from 1 to 100. We also consider the simulation data (Data 1 and Data 2) in the previous section. We randomly split each dataset into training, test, and calibration sets, with proportions 0.3, 0.2, and 0.5 respectively. On the training dataset, a random forest model is fitted as the ``black-box'' predictive model $\hat\mu$. For each dataset, we repeat the entire procedure 5 times, including splitting the data and forming the intervals. We report the average metrics taken over these trials}.

We compare with standard split-conformal prediction \citep{lei2014distribution}, locally-weighted conformal inference {(Locally Weighted)} \citep{lei2014distribution}, and conformalized quantile regression (CQR) \citep{romano2019conformalized}. Locally-weighted conformal inference is similar to our method in that it relies on a locally varying model for the conformity score $\hat{\sigma}:\mathcal{X}\to\mathbb{R}$, which is fit on the training data. We use CART as well as a random forest model for $\hat{\sigma}$ in our experiments. {When applying CQR, we likewise employ two {versions}, where the model for the conditional upper and lower quantiles is a single tree and a forest with 100 trees respectively}. Recall that our setting considers the inaccessibility of the original training data used for training the regression model. Since all of these benchmark algorithms need access to the training data set (or some other set other than the calibration set) to train their auxiliary or grouping functions, we used a random half of the calibration split to train the auxiliary functions and the remaining half for calibrate. \begin{table}
\centering
\resizebox{\columnwidth}{!}{
\begin{tabular}{|l|ccc|cccccc|cccc|cccc|}
\hline
 & \multicolumn{3}{c|}{Split Conformal} & \multicolumn{6}{c|}{Conformal Tree} & \multicolumn{4}{c|}{Locally Weighted (Tree)} & \multicolumn{4}{c|}{CQR (Tree)} \\
 & Width & Cov. & ISL & Width & Cov. & ISL & P.B. & Opt. $m$ & Opt. $K$ & Width & Cov. & ISL & P.B. & Width & Cov. & ISL & P.B. \\
\hline
Data 1 ($n=499$) & 4.29 & 0.92 & 5.28 & \textbf{3.5} & 0.86 & 4.66 & \textbf{0.59} & 20 & 4 & 3.52 & 0.88 & \textbf{4.48} & 0.54 & 4.84 & 0.93 & 5.77 & 0.0 \\
Data 2 ($n=499$) & \textbf{2.93} & 0.91 & 6.27 & 2.96 & 0.86 & \textbf{3.98} & \textbf{0.89} & 20 & 8 & 3.53 & 0.91 & 4.39 & 0.54 & 8.4 & 0.94 & 10.23 & 0.0 \\
bike ($n=5443$) & 1.37 & 0.9 & 2.17 & 1.33 & 0.88 & 2.06 & 0.61 & 20 & 64 & \textbf{1.22} & 0.89 & \textbf{1.52} & \textbf{0.67} & 2.67 & 0.89 & 4.15 & 0.0 \\
bio ($n=22864$) & 1.7 & 0.9 & 2.27 & 1.66 & 0.89 & \textbf{2.24} & \textbf{0.67} & 20 & 64 & \textbf{1.63} & 0.9 & 2.25 & 0.57 & 3.37 & 0.89 & 3.93 & 0.0 \\
community ($n=997$) & 1.9 & 0.89 & 2.89 & 1.78 & 0.89 & 2.8 & 0.55 & 20 & 32 & \textbf{1.77} & 0.9 & \textbf{2.46} & \textbf{0.59} & 3.6 & 0.9 & 5.09 & 0.0 \\
concrete ($n=514$) & 0.68 & 0.92 & 0.85 & \textbf{0.61} & 0.89 & \textbf{0.79} & \textbf{0.72} & 20 & 8 & 0.7 & 0.93 & 0.84 & 0.21 & 1.47 & 0.91 & 1.78 & 0.0 \\
homes ($n=10807$) & \textbf{0.62} & 0.91 & 1.18 & 0.64 & 0.9 & 1.03 & \textbf{0.7} & 20 & 32 & 0.7 & 0.9 & \textbf{0.95} & 0.63 & 1.19 & 0.9 & 2.18 & 0.0 \\
star ($n=1080$) & \textbf{0.17} & 0.89 & \textbf{0.22} & \textbf{0.17} & 0.87 & 0.23 & \textbf{0.67} & 20 & 4 & \textbf{0.17} & 0.9 & 0.23 & 0.21 & 0.29 & 0.91 & 0.33 & 0.0 \\
\hline
\end{tabular}
}
\caption{Empirical results on real and synthetic data for tree-based localization methods for $\alpha=0.1$. We display the average over five trials of the width, coverage, interval score length (ISL) and proportion of test points with intervals smaller than split conformal (P.B.). We display the calibration set size as $n$. Hyperparameters $m$ and ${K_{\rm max}}$ are tuned for each dataset and for each method by minimizing ISL on held-out validation data over a coarse grid.} \label{table:tuned-results}
\end{table}

{In addition, for all methods, we further partition one fifth of the remaining calibration set for each method to be used to tune hyperparameters. We tune the minimum samples per leaf, the maximum leaf nodes, and the number of trees (when applicable) for each method, by selecting the combination of hyperparameters that yield the minimal interval score loss \citep{Gneiting01032007} when evaluated on the held-out tuning data.} {Interval score loss (ISL) with threshold $\alpha$ is a proper scoring rule for prediction intervals that is minimized by the true $\alpha/2$ and $1-\alpha/2$ conditional quantiles \citep{Gneiting01032007} and has been used to evaluate the performance of conformal prediction methods \citep{rossellini2023integrating}. The ISL for observation $y$ given the $1-\alpha$ prediction interval $[\ell, u]$ is given by $(u - \ell) + \frac{2}{\alpha}(\ell - y) \mathbb{I}\{y < \ell\} + \frac{2}{\alpha}(y - u) \mathbb{I}\{y > u\}$.}
In Appendix \ref{sec:add_vis}, we also compare with other similar methods involving the use of regression trees and forests such as \citet{izbicki2022cd}, \citet{cabezas2025regression}, \citet{martinez2024identifying}, which find data-driven groups on additional datasets as well.

\begin{table}
\centering
\resizebox{\columnwidth}{!}{
\begin{tabular}{|l|ccc|cccc|cccc|cccc|}
\hline
 & \multicolumn{3}{c|}{Split Conformal} & \multicolumn{4}{c|}{Conformal Forest} & \multicolumn{4}{c|}{Locally Weighted (Forest)} & \multicolumn{4}{c|}{CQR (Forest)} \\
 & Width & Cov. & ISL & Width & Cov. & ISL & P.B. & Width & Cov. & ISL & P.B. & Width & Cov. & ISL & P.B. \\
\hline
Data 1 ($n=499$) & 4.2 & 0.89 & 5.8 & 4.12 & 0.95 & 4.57 & 0.51 & 3.71 & 0.92 & 4.53 & 0.48 & \textbf{3.67} & 0.95 & \textbf{4.18} & \textbf{0.65} \\
Data 2 ($n=499$) & \textbf{2.81} & 0.87 & 6.22 & 3.05 & 0.86 & \textbf{4.09} & \textbf{0.91} & 3.1 & 0.86 & 5.31 & 0.71 & 4.06 & 0.86 & 5.38 & 0.76 \\
bike ($n=5443$) & 1.38 & 0.9 & 2.15 & 1.14 & 0.93 & \textbf{1.28} & \textbf{0.72} & \textbf{1.1} & 0.9 & 1.64 & 0.7 & 1.46 & 0.89 & 1.62 & 0.49 \\
bio ($n=22864$) & 1.71 & 0.9 & 2.27 & 1.52 & 0.92 & \textbf{1.95} & 0.62 & \textbf{1.43} & 0.9 & 2.02 & \textbf{0.69} & 1.74 & 0.89 & 1.98 & 0.33 \\
community ($n=997$) & 1.84 & 0.88 & 2.97 & \textbf{1.64} & 0.9 & \textbf{2.39} & \textbf{0.66} & 1.66 & 0.91 & 2.41 & 0.63 & 1.67 & 0.9 & \textbf{2.39} & 0.59 \\
concrete ($n=514$) & 0.72 & 0.93 & 0.85 & 0.7 & 0.92 & 0.84 & 0.66 & \textbf{0.66} & 0.93 & \textbf{0.81} & \textbf{0.78} & 1.13 & 0.9 & 1.25 & 0.0 \\
homes ($n=10807$) & 0.63 & 0.91 & 1.25 & 0.64 & 0.93 & \textbf{0.85} & 0.69 & \textbf{0.57} & 0.9 & 0.89 & \textbf{0.74} & 0.76 & 0.9 & 1.01 & 0.6 \\
star ($n=1080$) & \textbf{0.18} & 0.92 & \textbf{0.2} & 0.19 & 0.94 & \textbf{0.2} & 0.23 & \textbf{0.18} & 0.92 & \textbf{0.2} & \textbf{0.46} & \textbf{0.18} & 0.92 & 0.21 & 0.36 \\
\hline
\end{tabular}
}
\caption{Empirical results on real and synthetic data for forest-based localization methods for $\alpha=0.1$. We display the average over five trials of the width, coverage, interval score length (ISL) and proportion of test points with intervals smaller than split conformal (P.B.). We display the calibration set size as $n$. Hyperparameters $m$ and ${K_{\rm max}}$ are tuned for each dataset and for each method by minimizing ISL on held-out validation data. We use 100 trees for all forest-based methods.} \label{table:tuned-results2}
\end{table}

The results for $\alpha=0.1$ are displayed in Table \ref{table:tuned-results}, which clearly displays the tightening effect of Conformal Tree. It shows that the average interval widths and interval score loss values are competitive with the other tree-based methods. Moreover, Conformal Tree has the highest proportion of test points with intervals smaller than split conformal on all datasets except two (\texttt{bike} and \texttt{community}). Tables \ref{table:tuned-results} also shows the average empirical coverage on the held-out test data for each method. We observe that each method {nearly} achieves the desired coverage level of 0.9. While our method shows slight undercoverage relative to the nominal level $1-\alpha$, it is still within our theoretical bounds of $1-\alpha-\delta$. {In Table \ref{table:tuned-results2}, we also display the results of a variant of our method that uses a random forest of 100 CART trees (with bagging and random feature selection) for partitioning. We aggregate these prediction sets using a majority vote procedure (See Remark \ref{rm:forest}). We find this method to yield the lowest interval score length among forest-based procedures on the majority of datasets in our comparison, excluding Data 1 and \texttt{concrete}. Recall that this procedure only promises coverage up to $1-2\alpha-2\delta$, though we see coverage above 0.9 in all but one of our datasets. We note that this method sacrifices interpretability in exchange for a more flexible model for the conditional conformity score.}

\section{Large Language Model Uncertainty Quantification}\label{sec:LLM}

  {In this section, we apply Conformal Tree to quantify the uncertainty in the output from a black-box large language model (LLM) for two downstream classification tasks. Conformal Tree uses a set of calibration data in order to partition the input space into regions where the conformity score is relatively homogenous, and yields conditional conformal guarantees in each one. These regions do not need to be prespecified by the user, which can be especially useful in cases where it is difficult for an uninformed user to designate a reasonable partition \emph{a priori}. We consider predicting the U.S. state of a legislator based on a measure of their ideology in voting for bills, as well as a prediction of skin disease from a list of clinical symptoms. The latter task is meant to emulate a user asking a language model for medical advice, which we do not advocate for but rather view as a domain in which locally adaptive uncertainty quantification would be paramount. We view the LLM models as a deterministic model (with the temperature set to 0), so they provide an actual fit as opposed to prediction involving noise. For classification, we use the complement of the predicted class probability as the conformity score, {as in \citep{angelopoulos2020uncertainty,romano2020classification,sadinle2019least}}. Our adaptive approach enables  generative  LLM prediction sets to stretch and shrink based on the informativeness of the prompt.

\subsection{Classifying States of Legislators using an LLM}\label{sec:legislators}

Large language models have become extremely popular generative AI tools, allowing users to generate completions of text or to generate responses to questions or instructions \citep{long2022training}. In the current era of AI as-a-service, many performant LLMs are hosted on servers, to which the user submits a query and receives either next token probabilities or the sampled subsequent tokens themselves. In this case, the user does not have access to the models' training data, and cannot hope to apply most typical methods for local adaptivity. Trained on vast quantities of training data, state-of-the-art LLMs has been demonstrated to have an {impressive} ``understanding'' of the U.S. political atmosphere, proving capable of assigning scores or pairwise comparisons between perceived ideology of legislators that correlates highly with measures estimated from their voting history \citep{wu2023large, ohagan2024measurement}. We posit a related question in this work, questioning how well an LLM can predict the state in which a legislator serves from only their DW-NOMINATE score \citep{poole1985spatial}, a two-dimensional representation of their ideology based on voting data.

Thus, we consider a black-box model $f:\mathbb{R}^{2}\to \mathcal{S}^{50}$, where $\mathcal{S}^{50}$ denotes a $50$-dimensional simplex, or a probability measure over U.S. states. In our case, $\hat{f}$ is the resulting probabilities assigned by an LLM when conditioning on a context that gives a specific DW-NOMINATE score, and asks where a legislator with that score would be likely to be employed. While the relationship between DW-NOMINATE and state could be alternatively studied using empirical data, the predictions from an LLM represent something else entirely-- rather than the relationship that manifests itself in historical data, they represent the ``communal opinion'', or ``Zeitgeist'', about the relationship between DW-NOMINATE, as distilled through the lengthy training process of the language model. However, social scientists may be weary of how much they can trust the output of the LLM for such a specific downstream task, given its black-box nature, which motivates a practitioner to look for conformal sets of U.S. states that are guaranteed to contain the true state of the legislator with a desired probability level.

\begin{figure}[!t]
    \centering
    \includegraphics[width=\textwidth]{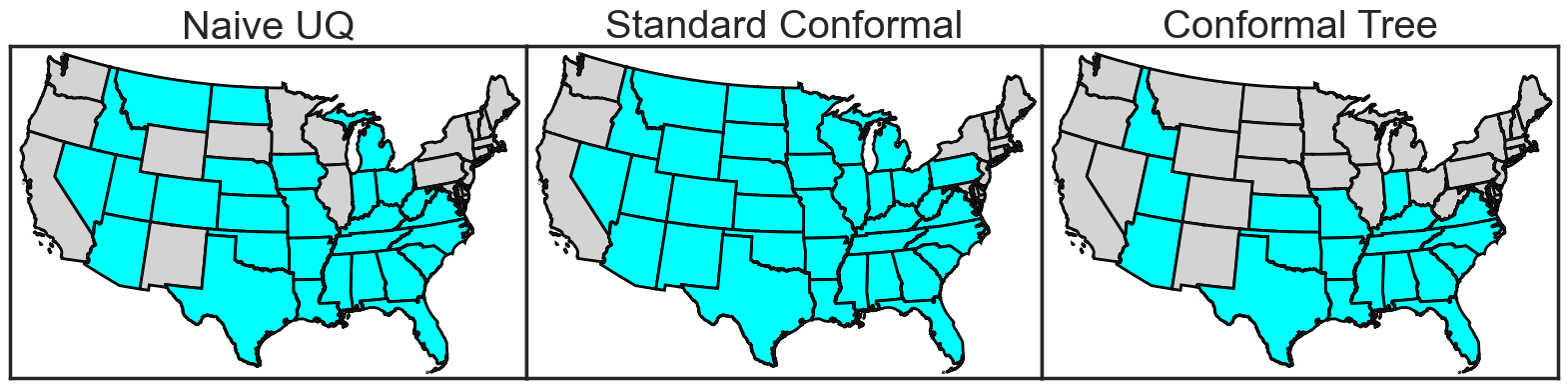}
    \caption{Example of conformal sets for a test point with DW-NOMINATE coordinates $(0.647,0.033)$, corresponding to a representative from Indiana. Conformal tree is able to shrink the size of the conformal set for test points in this region of DW-NOMINATE space.}
    \label{fig:gpt-example}
\end{figure}

For this classification task, we consider the conformity score $S(\bm{x},y)=1-f(\bm x)_y$, where $\bm{x}\in\mathbb{R}^2$ denotes a DW-NOMINATE score and $y\in [50]$ denotes a U.S. state index. A conformal prediction set is defined by $\wh{C}(\bx)=\{k:f(\bx)_k \geq 1-q\}$, where $q$ is the $1-\alpha$ quantile of $\{1-f(\bx)_y\}$. The \emph{mean prediction set size} of a conformal rule is given by $\sum_{i\in I_{\mathrm{test}}}|\wh{C}(\bx_i)|/|I_{\mathrm{test}}|$, which measures the average size of conformity sets over each test point. We can also look at the class-wise average set size for any state $s$ given by $\sum_{i\in I_{\mathrm{test}}}\mathbb{I}\{y_i=s\}\cdot |\wh{C}(\bx_i)|/\sum_{i\in I_{\mathrm{test}}}\mathbb{I}\{y_i=s\}$. For these metrics, a lower value is better.

We aggregated historical data on U.S. legislators in congress from the $100^{\textrm{th}}$ to $117^{\textrm{th}}$ congress, which contains 1809 unique legislators. We split this dataset into 80\% of legislators to be used as calibration data, and 20\% to evaluate the performance of our technique. As the underlying prediction model $f:\mathbb{R}^2\to \mathcal{S}^{50}$ is a specific downstream task of a pre-trained generative model, we do not have access to the training data. For each legislator, our data includes their DW-NOMINATE score and the U.S. state that they serve. For each legislator, we feed their DW-NOMINATE score to the language model and obtain U.S. state probability assignments. GPT-4o correctly identifies the correct state 8.3\% of the time, compared to 9.4\% for a random forest classifier, and 3.7\% for a random classifier that guesses with probability equal to that of the underlying dataset.

We apply our Conformal Tree methodology, partitioning the two-dimensional space of DW-NOMINATE scores using a dyadic tree that captures local variations in the conformity score. Our conformal prediction sets have different probability thresholds in each tree leaf, causing local variations in the sizes of the conformal sets that reflect the variation in the calibration data. {This type of variation reflects quality (informativeness) of  prompts for uncertainty quantification.} {In our experiments, we set $m=200$ and ${K_{\rm max}}=7$, where seven is the largest value possible within the restriction ($K_{\rm max}< (n+1)/m$).}

\begin{figure}[!t]
    \centering
    \includegraphics[width=0.56\textwidth]{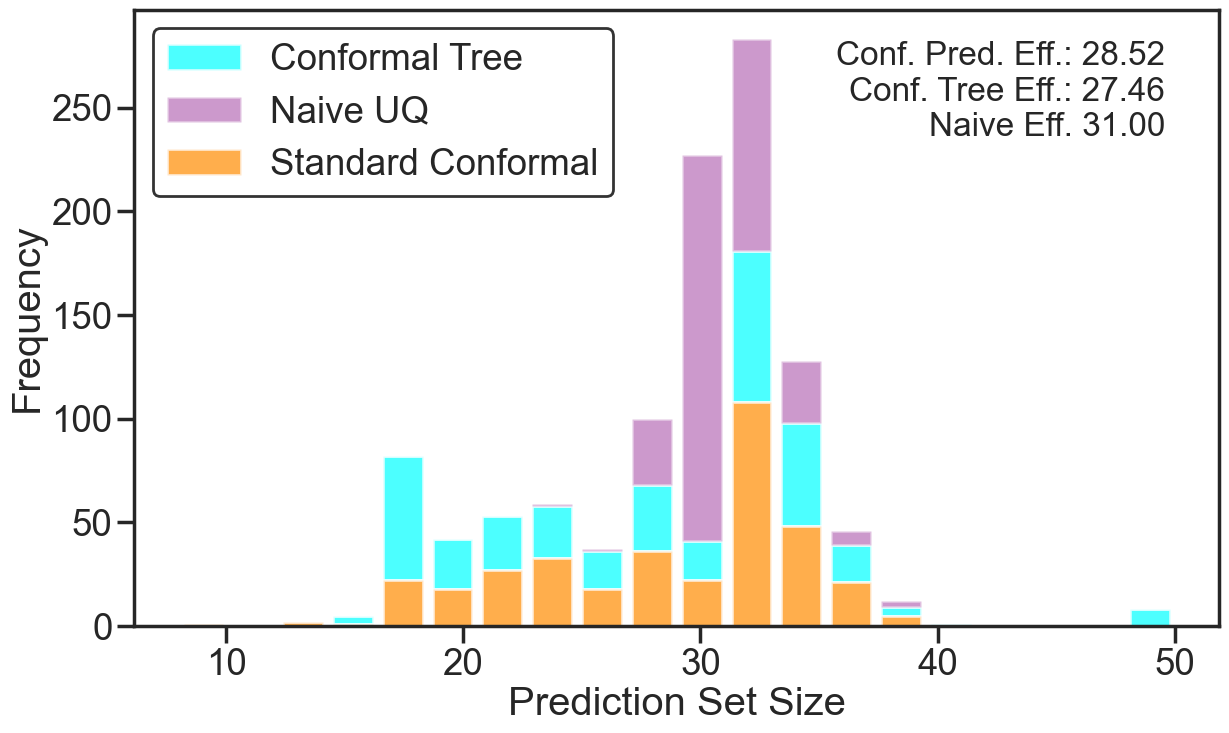}
    \includegraphics[width=0.41\linewidth]{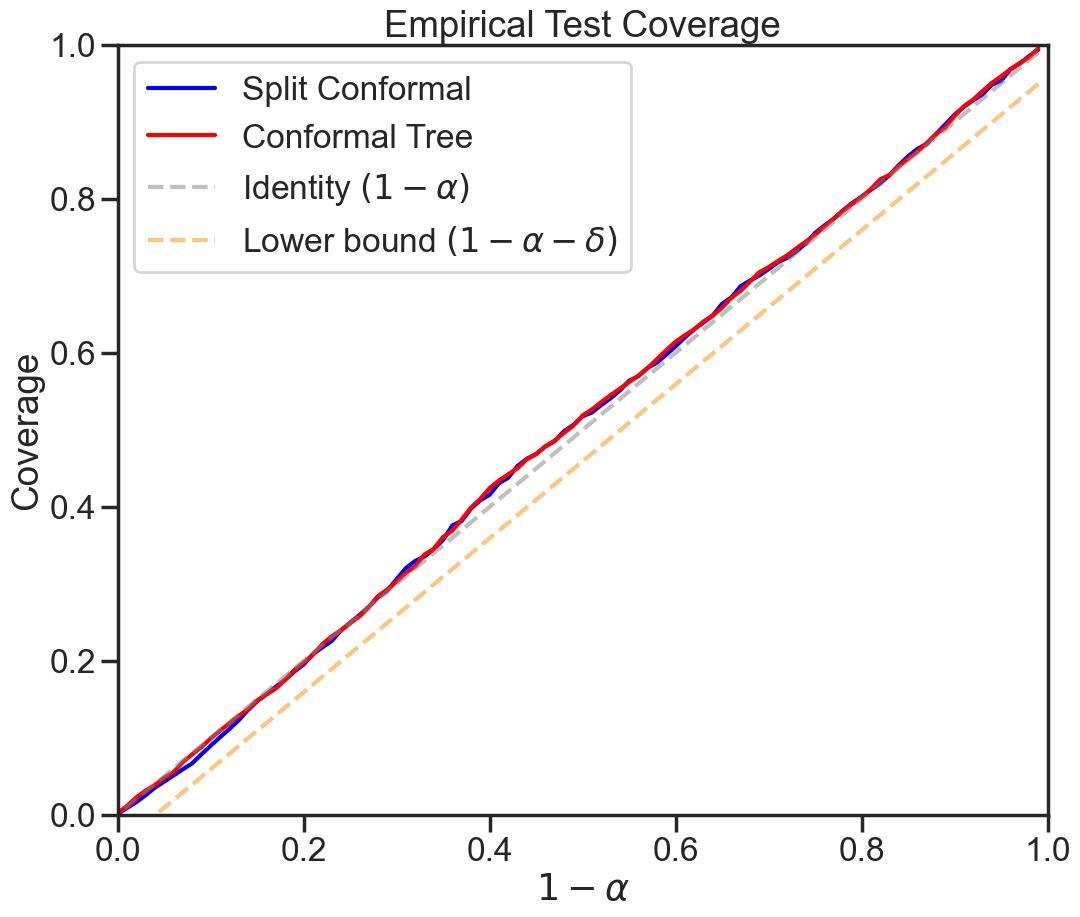}
    \caption{(Left) Histogram of conformal set sizes for U.S. state classification from ideology using GPT-4o using Conformal Tree, standard split conformal prediction, and naive UQ for $\alpha=0.2$. (Right) Empirical test coverage of split-conformal prediction and Conformal Tree on heldout test data, for varying $\alpha$ levels and $m=200$. Coverage is on heldout data averaged over ten folds. }
    \label{fig:gpt-4o-efficiency}
\end{figure}

Figure \ref{fig:gpt-example} showcases an example of 80\% prediction sets for an example test point, in this case corresponding to a U.S. representative from {Indiana}. Both standard split-conformal prediction and Conformal Tree yield 80\% prediction sets that contain Indiana for this test point, but the size of the the set decreases from 35 states for split-conformal prediction to 19 states for Conformal Tree. This is because in this region of DW-NOMINATE space, we are relatively more confident in the output of the black-box model, allowing us to locally increase the threshold $q$ and obtain smaller conformal sets in this region. 

\begin{figure}[!t]
    \centering
    \includegraphics[width=\textwidth]{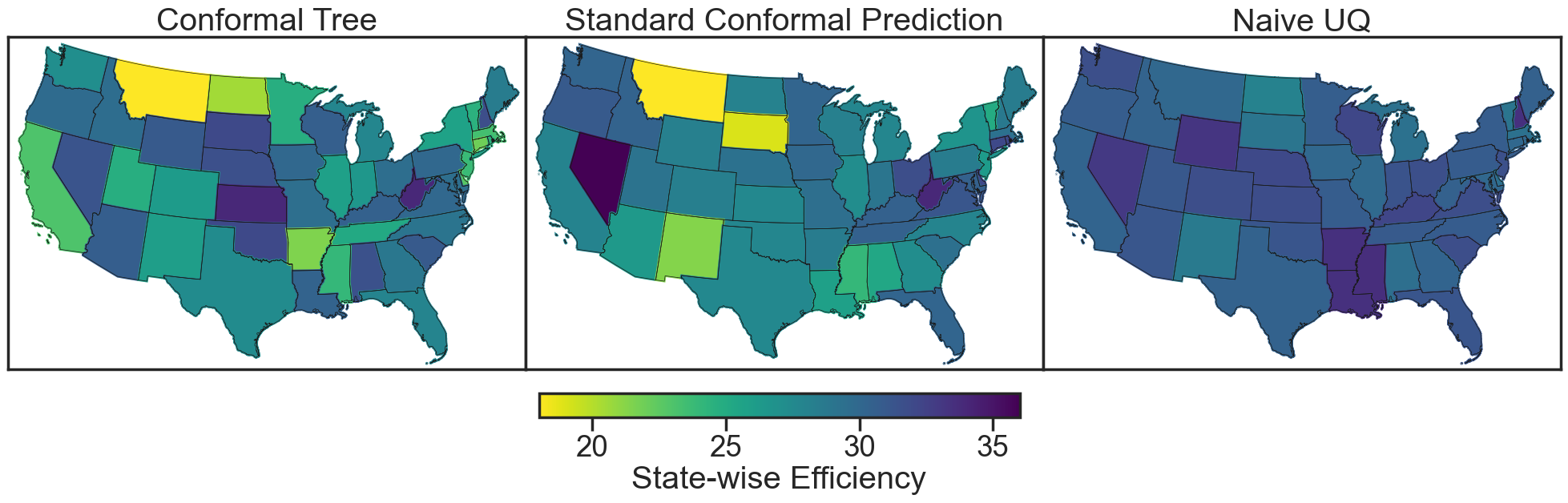}
    \caption{State-wise average prediction set sizes for each U.S. state. Lower is better. Conformal tree reduces the median state-wise prediction set size by 4.0 when compared to standard split-conformal prediction.}
    \label{fig:statewise-eff}
\end{figure}

In Figure \ref{fig:gpt-4o-efficiency} (a), we fix a threshold of $\alpha=0.2$ and use Conformal Tree as well as standard split-conformal prediction to obtain 80\% prediction sets for the state of every legislator in the test set. We plot histograms of the sizes of these prediction sets, observing that Conformal Tree yields a multimodal distribution over sizes, yielding more examples with smaller prediction sets. In fact, on 89.2\% of test examples, Conformal Tree's prediction set is the same size or smaller than that of standard split-conformal prediction. We also observe that the average prediction set size for Conformal Tree is lower, meaning that its average prediction set size is lower by more than one state. Its median prediction set size is lower by four states. Figure \ref{fig:gpt-4o-efficiency} (b) shows that Conformal Tree maintains the desired coverage level for any $\alpha\in(0,1)$. In Figure \ref{fig:statewise-eff}, we plot the state-wise average prediction set sizes, representing the average conformal set size for test legislators from that particular state. While Conformal Tree is less efficient in more difficult to classify states, like Kansas and Maine, it is relatively more efficient in many states, with an improvement in the median state-wise prediction set size of 0.3. 
\subsubsection{Naive Uncertainty Quantification with ChatGPT}
\label{sec:naive}
We also compare our calibrated conformal prediction sets against naive prediction sets based on ChatGPT's own assessment of its uncertainty. In this case, we set the temperature hyperparameter to one, meaning that the language model's response sample each subsequent token from the softmax distribution rather than deterministically choosing the highest probability token. This allows us to sample distributions on labels for each test point. For each sampled distribution $q\in\mathcal{S}^L$, {where $L$ is the cardinality of the set of labels}, we compute a naive prediction set by including the largest probability items until at least $1-\alpha$ is covered. That is, we can define a threshold 
$
\tau=\inf\{t:\sum_{l:q_l\geq t}q_l\geq 1-\alpha\}
$
and define the prediction set corresponding to $q$ as
\begin{equation}
    \tilde{C}(q) = \{l\in[L]:q_l \geq \tau\}\,.
\end{equation}

Because each individual sampled $q$ for a given test point involves randomness from the model, we can combine sampled distributions $q^1,\ldots,q^M$ for a given test point using a majority vote procedure
\begin{equation}
    \hat{C}=\{l\in[L]:\sum_{j=1}^M \mathbb{I}\{l\in\tilde{C}(q^j)\}\geq M/2\}\,.
\end{equation}
The conformal set $\hat{C}$ is a final prediction set for the given test point that combines sampled naive prediction sets that are based on the model's own quantification of its uncertainty. For our comparisons, we set $M=11$. 

In our previous example classifying U.S. states from DW-NOMINATE scores, the average size of these naive prediction sets for each method was 36.36 items, significantly larger than that of Conformal Tree (27.67).
 
\subsection{Conformalizing ChatGPT Diagnoses of Skin Diseases}\label{sec:skin_ex}


One significant yet controversial use of ChatGPT lies in healthcare, where people have the option to use it to obtain preliminary diagnoses of their medical conditions based on physical descriptions and perceived symptoms \citep{Garg2023}. Previous studies have found ChatGPT to be effective at creating shortlists of possible diagnoses based on clinical vignettes from case reports \citep{shieh2024assessingChatGPT} and to be an effective self-diagnostic tool for common orthopedic diseases \citep{Kuroiwa2023}. Uncertainty quantification of these black-box generative models is important for its potential future use in medical diagnostics, but also addresses the current reality that people may use these tools for self-diagnosis, despite potential risks. 

While early investigations into the diagnostic capabilities of LLMs like ChatGPT are underway, it is crucial for users to have a means to quantify the uncertainty associated with these diagnoses. Given that these black-box models can be difficult to understand and trust, Conformal Tree provides conformal prediction sets for ChatGPT as diagnoses, with locally varying thresholds to represent different levels of confidence for different groups of patient characteristics. This allows users to have a statistical guarantee on the coverage of the generated diagnoses, offering a layer of reliability in a domain where precision would be paramount.

\begin{figure}
    \centering
    \begin{tikzpicture}

    \node[draw, text width=0.25\linewidth, align=left] (patient) at (0, 0) {
        \textbf{Patient Covariates} \\
        Age: 52 \\
        Redness: 3/3 \\
        Itchiness: 3/3 \\
        On knees and elbows
    };

    \draw[->] (patient) -- (3.5, 0);

    \node[draw, text width=0.25\linewidth, align=left, fill=green!10] (llm) at (6, 0) {
        \textbf{LLM} \\
        I am \textbf{52} years old and 
        have a \textbf{red}, \textbf{very itchy} 
        rash \textbf{all over my knees 
        and elbows}, what is 
        my diagnosis?
    };

    \draw[->] (llm) -- (9.5, 0);

    \node[draw, text width=0.2\linewidth, align=left, fill=purple!10] (diagnosis) at (11.5, 0) {
        \textbf{Diagnosis} \\
        Psoriasis: 0.7 \\
        Seborrheic \\
        dermatitis: 0.25 \\
        Chronic \\
        dermatitis: 0.05
    };

    \end{tikzpicture}
    \caption{Obtaining diagnoses by asking a large language model. Exact list of covariates and prompt used are available in Section \ref{sec:full_skin_detail}.}
    \label{fig:flowchart}
\end{figure}

We consider the erythemato squamous disease diagnosis task based on patient characteristics and physical descriptions \citep{misc_dermatology_33}, which consists of 366 observations, 35 features. Erythemato squamous disease is a class of common skin diseases, that can manifest as rashes, redness, or scaly patches in the skin. The disease is primarily categorized in six different forms. We consider an example of using generative AI for self-diagnosis by applying GPT-4o to give diagnoses for the type of erythemato squamous disease present based only on physical characteristics that a patient could report from an inspection. The process by which self-diagnoses are obtained from the language model is outlined in Figure \ref{fig:flowchart}, with full details given in Section \ref{sec:full_skin_detail}. We used the 12 features from patients, comprised of a mixture of continuous, ordinal, and binary variables. We ignore the remaining 23 histopathological features, as the purpose of this experiment is to emulate a user describing physical symptoms to a language model. 

\begin{figure}[!t]
        \includegraphics[width=\linewidth]{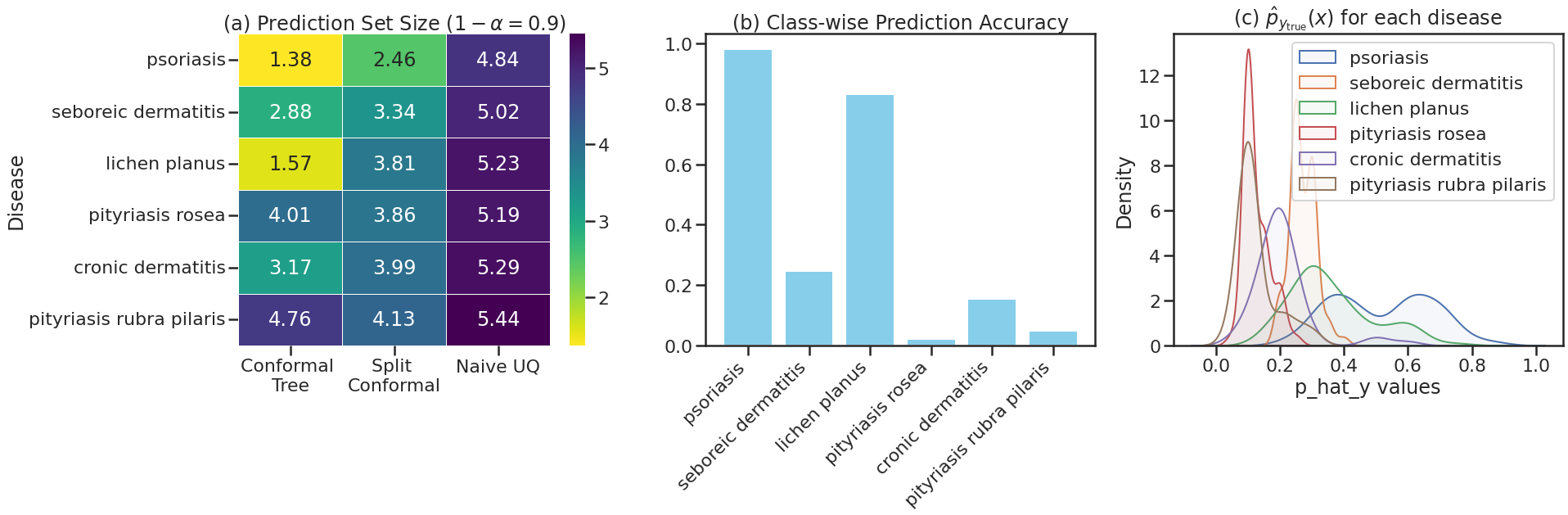}
    \caption{(a): Comparison of prediction set sizes for erythemato squamous disease using GPT-4o. 
    (b) Class-wise prediction accuracy of ChatGPT. (c) The prediction confidence of ChatGPT on the true label. The ChatGPT performance order among labels in (b) and (c) aligns with the tightness order in (a) of Conformal Tree results.} 
    \label{fig:combined_derm}
\end{figure}
On a heldout set of 20\% of the data, the disease which ChatGPT assigns highest probability is the correct disease 51\% of the time, which is significantly better than a random classifier baseline, which guesses each class with probability proportional to its underlying probability in the data, which would achieve only a 1/6 accuracy. However, as one might expect, ChatGPT is outperformed by a model specifically trained for this task, as a random forest model achieves 85\% accuracy on the same heldout test data. We applied Conformal Tree to obtain 90\% prediction sets with $m=10$ and ${K_{\rm max}}=20$, as well as the Split Conformal prediction and Naive UQ. For conformal predictions, we used 60\% of the data to calibrate our method, and heldout the remaining 40\% in order to test it. }{As the prediction probability vectors from ChatGPT tend to be limited to a certain set of values, we added a tiny random perturbation of uniform(0.0001, 0.001) to break the tie and normalized (see, \citet{tibshirani2019conformal}).}

Figure \ref{fig:combined_derm} (a) shows the class-wise average prediction set sizes for Conformal Tree compared to standard split-conformal prediction. To explain the interpretability of this result, we also show additional information in Figures \ref{fig:combined_derm} (b) and (c), which represent the performance of ChatGPT in terms of prediction accuracy and confidence, respectively. Now, we can see that the predicted set size in Figure \ref{fig:combined_derm} (a) is smallest for psoriasis and then for lichen planus, seboreic dermatitis, and chronic dermatitis, which aligns the order of the prediction accuracy and confidence observed in Figure \ref{fig:combined_derm} (b) and (c).  In terms of overall prediction efficiency, we draw the frequency plot in Figure \ref{fig:combined_derm_additional} in Section \ref{sec:add_vis}, which shows the frequency of a single element set is the highest for Conformal Tree.


\begin{figure}
    \centering
    \includegraphics[width=\linewidth]{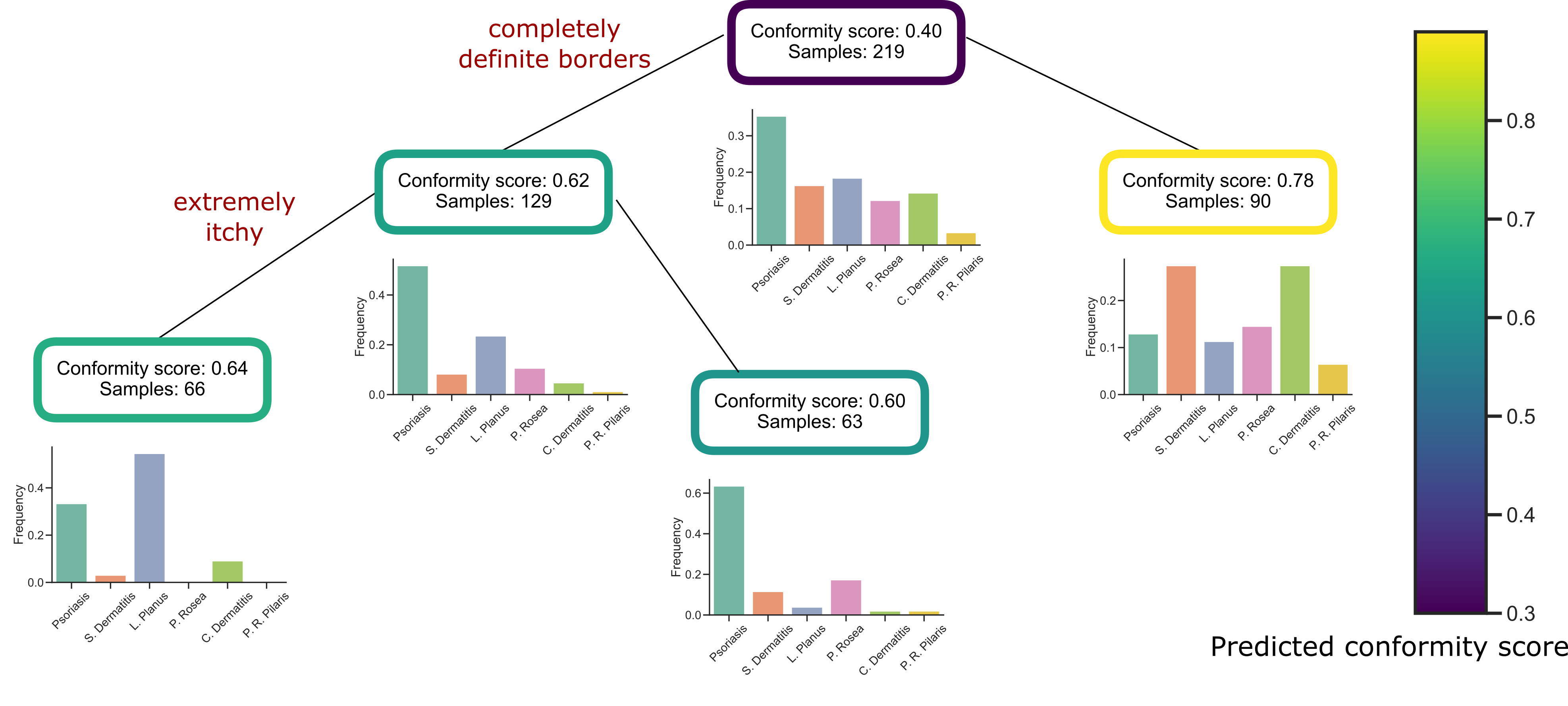}
    \caption{Conformal Tree's full structure on the skin disease conformal data, with $m=50$ and $K_{max}=4$. Node colors correspond to the mid-range of conformity scores (of calibration data) within the node. Below each node, we plot the empirical distribution of calibration data corresponding to each disease within that node. The tree and partition correspond to that of Figure \ref{fig:boxwise_vis} in Section \ref{sec:add_vis}. 
    }
    \label{fig:derm-tree-structure2}
\end{figure}

{For a more interpretable presentation of the results, we intentionally set $m=50$ to obtain a tree with a smaller number of partition boxes.} In Figure \ref{fig:derm-tree-structure2}, we visualize each level of Conformal Tree's structure on the skin disease calibration data. This plot emphasize the interpretability of our method by showing the precise splits that determine each leaf node, in which group conditional coverage is achieved. In Figure \ref{fig:boxwise_vis} in Section \ref{sec:add_vis}, we visualize the empirical coverage and threshold for decreasing $\alpha$ on each of the resulting partitions, obtained by applying Conformal Tree on one realization of the split of the data. The result is very interpretable. In Box 3, the threshold values of our method (orange color) are significantly smaller than those of split conformal prediction. Interestingly, Box 3 does not contain any data points of pityriasis rosea and pityriasis rubra pilaris, for which the ChatGPT shows the lowest performance. On the other hand, in Box 1, the thresholds of our method are set large even for the largest $\alpha$ (and increasing slightly as $\alpha$ decreases), as the pityriasis groups are heavily present in this box. The estimated coverage in Figure \ref{fig:boxwise_vis} can be improved if we had more test data points. It is implied by the dramatic difference in Figure \ref{fig:emp-coverage-test} (a) and (b) in Section \ref{sec:add_vis} that coverage estimation requires a large number of test data points. { In Figure \ref{fig:emp-coverage-test}, we show empirical test coverage for split-conformal and Conformal Tree for varying $1-\alpha$ values in the GPT diagnosis example. 
When the size of data is artificially increased by multiple random splits (see Section \ref{sec:add_vis} for details), the estimation quality improves. Note that, however, even without multiple splits (Figure \ref{fig:emp-coverage-test} (a)), the estimated coverage for conformal tree achieves the desired coverage level across the values of $\alpha$, staying within the theoretical bounds.} Note that our methodology does not guarantee class-wise coverage for each true class but box-wise coverage. For the tree structure of $m=20$, we also provide visualization of the interpretable tree-splitting process in Figure \ref{fig:derm-tree-structure} in Section \ref{sec:add_vis}.

\section{Concluding Remarks}\label{sec:concluding}

Conformal prediction serves to provide some statistical clarity in the otherwise uncharted waters of modern black-box deep learning. Conformal tree allows for local adaptation and conditional coverage guarantees based on a self-grouping procedure that is robust to the inclusion of an additional {i.i.d.} test point. This can be a valuable tool for quantifying uncertainty in black-box learning applications where the model  outputs a single point prediction using training data that has   been obscured from the end user. Future work may include extending the framework beyond trees to other classes of calibration models.

\bibliographystyle{abbrvnat}
\typeout{}
\bibliography{ref}


\newpage
\setcounter{page}{1}

\begin{appendices}

\begin{abstract}
This supplementary material contains the proof of all results as well as auxiliary lemmata in Section \ref{sec:proofs}. Section \ref{sec:tighter} includes a description of a modified algorithm that achieves tighter theoretical  guarantees at the cost of making computation scale linearly in the number of test points. Section \ref{sec:hyp_sense} conducts sensitivity analyses for the main hyperparameters $K_{\rm max}$ and $m$. In Section \ref{sec:exp-info}, we provide implementation details for the experiments and simulations present in the article, including data sources and preprocessing information, hyperparameter choices, and LLM prompt information. We also include additional experiments and visualizations referred to in the text in Section \ref{sec:add_vis}.
\end{abstract}

\startcontents[appendix]
\printcontents[appendix]{l}{1}{\section*{Supplementary Material Table of Contents}}

\section{Proofs of Theorems}
\label{sec:proofs}

First, we provide a definition of the rank of a random variable within a set of random variables that handles ties in such a way that for an exchangeable set of random variables, the rank of any particular random variable {(even discrete)} is uniformly distributed.

\begin{defn}
    Given a set of random variables $S_1,\ldots,S_n$, the rank of $S_i$ is defined by 
    \[
    R(S_i,\{S_1,\ldots,S_n\})= 1 + \sum_{j=1}^n \mathbb{I}\{S_j < S_i \} + U_i
    \]
    where $U_i\sim\mathrm{Unif}\{0,1,2,\ldots,T\}$ where $T=\sum_{j\neq i} \mathbb{I}\{S_j=S_i\}$. In other words, in the case of ties, ranks are assigned uniformly at random between the possible positions within the tie.
\end{defn}

\subsection{Technical Lemmata}\label{sec:tech}
The following auxiliary results play a key role in proving the lemmata and theorems in the main body. The first lemma (Lemma \ref{lem:exch}) will be useful when we need an exchangeability conditioned on the fitted tree. The second lemma (Lemma \ref{lem:foundation}) establishes a conditional conformal guarantee when a fixed partition is given, which will be used as a building block for the cases later the partitions are estimated.
\begin{lemma}
\label{lem:exch}
Let $(X_1,S_1),\ldots,(X_{n+1},S_{n+1})$ be i.i.d. from any joint probability distribution, where $(X_1,S_1)$ is supported on a subset of $\mathcal{X}\times \mathbb{R}$. Consider a function $f:(\mathcal{X}\times \mathbb{R})^{n+1}\to E$, where $E$ is an arbitrary set. We say that $f$ is \emph{symmetric} in $[n+1]$ if for any element of the domain $(X_1,S_1),\ldots,(X_{n+1},S_{n+1})$ and permutation $\sigma:[n+1]\to[n+1]$, $
    f((X_1,S_1),\ldots,(X_{n+1},S_{n+1}))=f((X_{\sigma(1)},S_{\sigma(1)}),\ldots,(X_{\sigma(n+1)}, S_{\sigma(n+1)})).$ Define a summary statistic $T = f((X_1,S_1),\ldots,(X_{n+1},S_{n+1}))$ and  let $I(T)$ be a subset of $[n+1]$ that can depend on $T$. Then we have 
    
    (i) If $f$ is symmetric in $[n+1]$, then $Z_i\mid T$ for $i\in I(T)$ are exchangeable for any $I(T)$.

    (ii) Otherwise, there exists $I(T)$ for which $Z_i\mid T$ for $i\in I(T)$ are not exchangeable.
\end{lemma}

\begin{proof}
Denote $Z_i = (X_i,S_i)$ for notational convenience. We assume that all relevant random variables admit a probability density function, and denote these densities by $p$. First, consider when $f$ is symmetric in $[n+1]$. For any permutation $\sigma$ of $[n+1]$, 
    \begin{align*}
        p(Z_{\sigma(1)},\ldots,Z_{\sigma(n+1)}\mid T) &=p(T\mid Z_{\sigma(1)},\ldots,Z_{\sigma(n+1)}) p(Z_{\sigma(1)},\ldots,Z_{\sigma(n+1)}) /p(T)\\ 
        &= p(T\mid Z_{1},\ldots,Z_{n+1}) p(Z_{\sigma(1)},\ldots,Z_{\sigma(n+1)}) /p(T)\\ 
        &= p(T\mid Z_{1},\ldots,Z_{n+1}) p(Z_{1},\ldots,Z_{n+1})/p(T) \\ 
        &= p(Z_{1},\ldots,Z_{n+1}\mid T)
    \end{align*}
    where the second equality holds if and only if $f$ is symmetric in $[n+1]$, and the third equality uses exchangeability of $Z_1,\ldots,Z_k$. This implies that $Z_1\mid T,\ldots, Z_{n+1}\mid T$ are exchangeable if and only if $T$ is symmetric in $[n+1]$. For any subset $I(T)\subset[n+1]$, denote the included indices as $i_1,\ldots,i_k$. The previous display implies that
    \begin{equation}\label{eq:exch1}
        p(Z_{\sigma(i_1)},\ldots,Z_{\sigma(i_k)}\mid T)= p(Z_{i_1},\ldots,Z_{i_k}\mid T)
    \end{equation}
    meaning that $f$ being symmetric in $[n+1]$ implies $Z_i\mid T$ for $i\in I(T)$ are exchangeable. 

    If $f$ is not symmetric in $[n+1]$, let $I(T) = [n+1]$. Since the random variables $Z_1,\ldots,Z_{n+1}$ are not exchangeable given $T$, {\textit{(ii)}} is apparent.

\end{proof}

{The following lemma states for a predetermined partition of $\mathcal{X}$ that a split conformal set created separately by calibrating on data within each group provides conditional conformal guarantees \citep{lei2014distribution}. This lemma is included for completeness and will be used later in the proofs of other lemmata.}

\begin{lemma} \label{lem:foundation} Let $(X_1,Y_1),\ldots,(X_{n+1},Y_{n+1})$ be {i.i.d. samples from a joint probability measure $\mathbb{P}$}. Given a fixed partition $\mathfrak{X} = \{\mathscr{X}_1, ..., \mathscr{X}_K\}$, we have the group-conditional coverage guarantee 
\begin{equation}\label{eq:fixed_cond_lower}
    \mathbb{P}\{Y_{n+1}\in\wh{C}^{\mathfrak{X}}_{n,m}(X_{n+1})\mid X_{n+1}\in \mathscr{X}_k\}\geq 1-\alpha,
\end{equation} for every $k\in \{1,...,K\}$, where $\wh{C}^{\mathfrak{X}}_{n,m}(X_{n+1})$ is defined in \eqref{eq:conditional_band}. Additionally, provided that the conformity scores
have a continuous joint distribution, and when the least number of points in $\mathscr{X}_k$ is lower bounded by $m$, we have an upper bound 
\begin{equation}\label{eq:fixed_cond_upper}
\P\{Y_{n+1}\in \wh{C}^{\mathfrak{X}}_{n,m}(X_{n+1})| X_{n+1}\in \mathscr{X}_k\}\leq 1-\alpha+\frac{1}{m+1}.
\end{equation}
\end{lemma}

\begin{proof}
Denote by $\mathbb{I}$ the set of all possible subset of $[n]$. We can decompose 
\begin{align}
    \mathbb{P}\{Y_{n+1}&\in\wh{C}^{\mathfrak{X}}_{n,m}(X_{n+1})\mid X_{n+1}\in \mathscr{X}_k\} \nonumber\\= \sum_{I\in\mathbb{I}}~&\mathbb{P}\{Y_{n+1}\in\wh{C}^{\mathfrak{X}}_{n,m}(X_{n+1})\mid X_{n+1}\in \mathscr{X}_k, X_{I}\in \mathscr{X}_k, X_{[n]\setminus I}\not\in \mathscr{X}_k\}\nonumber\\ &\times \mathbb{P}\{X_{n+1}\in \mathscr{X}_k, X_{I}\in \mathscr{X}_k, X_{[n]\setminus I}\not\in \mathscr{X}_k\}.\label{eq:all_comb}
\end{align}
Note that when conditioned on an event  $A_I =\{ X_{n+1}\in \mathscr{X}_k, X_{I}\in \mathscr{X}_k, X_{[n]\setminus I}\not\in \mathscr{X}_k\}$, we have   exchangeability. To see this, denote $Z_i = (X_i,S(X_i,Y_i))$ for notational convenience and assume that all relevant random variables admit a probability density function, and denote these densities by $p$. Due to the i.i.d. assumption, we have exchangeability as
\begin{align}
    p(Z_{i_1},...,Z_{i_{|I|}} | A_I)& = p(Z_{i_1},...,Z_{i_{|I|}} | X_{i_1}\in \mathscr{X}_k,...,X_{i_{|I|}}\in \mathscr{X}_k) = \prod_{j=1}^{|I|} p(Z_{i_j}|X_{i_j}\in \mathscr{X}_k) \nonumber \\
        &=p(Z_{\sigma(i_1)},...,Z_{\sigma(i_{|I|})} | X_{i_1}\in \mathscr{X}_k,...,X_{i_{|I|}}\in \mathscr{X}_k)\nonumber\\
        &=    p(Z_{\sigma(i_1)},...,Z_{\sigma(i_{|I|})} | A_I),\label{eq:ground_for_exch}
\end{align}
where $i_1,...,i_{|I|}$ are the elements of $I$. Therefore, we have \[\mathbb{P}\{Y_{n+1}\in\wh{C}^{\mathfrak{X}}_{n,m}(X_{n+1})\mid X_{n+1}\in \mathscr{X}_k, X_{I}\in \mathscr{X}_k, X_{[n]\setminus I}\not\in \mathscr{X}_k\}\geq 1-\alpha\] since the prediction set $\wh{C}^{\mathfrak{X}}_{n,m}(X_{n+1})$ can be interpreted as standard split-conformal prediction restricted to the data in $\mathscr{X}_k$, where the conditional lower bound is obtained by a standard exchangeability argument (i.e. \citep{lei2014distribution}, Thm. 2.2). Combining this with \eqref{eq:all_comb} leads to desired results. Likewise, we have the upper bound by observing 
\begin{align*}
\mathbb{P}\{Y_{n+1}\in\wh{C}^{\mathfrak{X}}_{n,m}(X_{n+1})\mid X_{n+1}\in \mathscr{X}_k, X_{I}\in \mathscr{X}_k, X_{[n]\setminus I}\not\in \mathscr{X}_k\}&\leq 1-\alpha+\frac{1}{1+\sum_{i=1}^n\mathbb{I}\{X_i\in\mathscr{X}_k\}}  \\
&\leq 1-\alpha+\frac{1}{m+1}.
\end{align*}
\end{proof}

\subsection{Proof of Lemma \ref{lem:tree_same}}\label{pf:tree_same}
\begin{proof}

Lemma \ref{lem:tree_same} is based on sufficient conditions for the resulting partition to be identical when fit on $n$ datapoints rather than $n+1$ (partition unchangeability), and bounds the probability of such conditions. The partition resulting from $n+1$ observations differing from that $n$ observations requires one of the following conditions:
    \begin{enumerate}
        \item The set of candidate nodes changes with the exclusion of $(X_{n+1},S_{n+1})$.
        \item The criterion value of a candidate node changes with the exclusion of $(X_{n+1},S_{n+1})$.
    \end{enumerate}
    Note also that a tree implementation with a minimum improvement threshold for splitting also requires condition 2 to be satisfied for the resulting tree to be different. This is because the binary variable of a split candidate's criterion value being above the threshold changing necessarily requires the criterion value of the node to change as well.
    For the first case, the only possibility is that new candidates are created if the exclusion of $(X_{n+1},S_{n+1})$ changes the number of datapoints inside a potential leaf from $m$ to $m-1$. For this to occur, it requires that exactly $m$ points are contained in the potential leaf node containing $X_{n+1}$. Denote by $\mathscr{X}_{k(X_{n+1})}$ the leaf of $\wh{\mathfrak{X}}(\mathcal{D}_{1:n+1})$ containing $X_{n+1}$. We define the event such that removing $X_{n+1}$ doest not change the candidacy (i.e. eligibility for splitting) of the node $\mathscr{X}_{k(X_{n+1})}$ as
    \begin{equation}\label{eq:cand_set}
    \mathcal{E}_{\mathrm{cand}} = \left\{\sum_{i=1}^{n+1} \mathbb{I}\{X_i\in\mathscr{X}_{k(X_{n+1})}\}\neq m\right\}.
\end{equation}
    In the second case, in order for the criterion value to change, $S_{n+1}$ must be the (uniquely) most extreme value (uniquely smallest or uniquely largest) inside a node. For this to be the case, it must be the uniquely most extreme value in its leaf node, as its leaf node is a subset of all nodes that contain it. We define the complement of this event by 
    \begin{equation}\label{eq:inliner}
         \mathcal{E}_{\mathrm{inlier}}=\left\{\min_{\substack{i:i \in [n] \\ X_i \in \mathscr{X}_{k(X_{n+1})}}} S_i \leq S_{n+1} \leq \max_{\substack{i:i \in [n] \\ X_i \in \mathscr{X}_{k(X_{n+1})}}} S_i\right\}\,.
    \end{equation} 
    Thus, 
    \begin{equation}\label{eq:tree_implication}
    \wh{\mathfrak{X}}(\mathcal{D}_{1:n}) \neq \wh{\mathfrak{X}}(\mathcal{D}_{1:n+1}) \implies \mathcal{E}^c_{\mathrm{inlier}}\cup \mathcal{E}_{\mathrm{cand}}^c\,.
    \end{equation}
    In tandem, these events guarantee that the criterion value in any node is unaffected by $S_{n+1}$, and that the candidacy of any potential node in the fitting process is unaffected by $X_{n+1}$, ensuring that the same partition will be obtained when $(X_{n+1},S_{n+1})$ is excluded. Using Lemma \ref{lem:bad_bound}, therefore, we have
    \begin{align*}
        \mathbb{P}\{\wh{\mathfrak{X}}(\mathcal{D}_{1:n})\neq \wh{\mathfrak{X}}(\mathcal{D}_{1:n+1})\} &\leq \mathbb{P}(\mathcal{E}^c_{\mathrm{\mathrm{inlier}}}\cup \mathcal{E}^c_{\mathrm{cand}}) \\ 
        &\leq \sum_{\wt{\mathfrak{X}}\in\mathbb{X}} \sum_{k=1}^{|\wt{\mathfrak{X}}|} \mathbb{P}(\mathcal{E}^c_{\mathrm{\mathrm{inlier}}} \cup \mathcal{E}^c_{\mathrm{cand}}\mid \wh{\mathfrak{X}}(\mathcal{D}_{1:n+1}) = \wt{\mathfrak{X}},  X_{n+1}\in\mathscr{X}_k) \\& ~~~~~~~~~~~~~~~\cdot \mathbb{P}\{\wh{\mathfrak{X}}(\mathcal{D}_{1:n+1})=\wt{\mathfrak{X}}, ~X_{n+1}\in\mathscr{X}_k\}\\
        &\leq \delta(n,m)\,.
    \end{align*}
    where $\mathbb{X}$ denotes the set of all possible partitions. 
\end{proof}

\begin{lemma}\label{lem:bad_bound}
    Under the conditions and notations of Lemma \ref{lem:tree_same}, define
    \begin{equation}
     \delta(n,m) := \frac{2}{m} + \exp\left\{-(2m-1)\left(1+\frac{1}{n}\right)\right\}
    \label{eq:delta_defn_new}
    \end{equation}
    Then for any $\wt{\frak{X}}\in\mathbb{X}$ and any $\mathscr{X}_k\in \wt{\frak{X}}\in\mathbb{X}$, we have 
    \begin{equation}
    \label{eq:tree_same_lemma_bound_new}
    \mathbb{P}\{\mathcal{E}^c_{\mathrm{inlier}}\cup\mathcal{E}^c_{\mathrm{cand}}\mid \wh{\mathfrak{X}}(\mathcal{D}_{1:n+1})=\wt{\frak{X}},X_{n+1}\in\mathscr{X}_k\} \leq \delta(n,m)\,.
    \end{equation}    
\end{lemma}
\begin{proof}
Note first that, as a result of Lemma \ref{lem:exch},  $\{S_i\mid \wh{\mathfrak{X}}_{n+1}:i\in [n+1]\}$ are exchangeable random variables as the conditioning statistic $\wh{\mathfrak{X}}_{n+1}$ is a symmetric function of the data $(X_1,S_1),\ldots, (X_{n+1}, S_{n+1})$. Applying \eqref{eq:ground_for_exch} in Lemma \ref{lem:foundation}, we have also that $\{S_i|\wh{\mathfrak{X}}_{n+1},X_{n+1}\in\mathscr{X}_k:i\in\mathscr{X}_k\}$, the conformity scores of data within the leaf node containing $X_{n+1}$  are conditionally exchangeable. We denote the number of data points in the leaf $k$ as $m_k:=|\{X_i:X_i\in\mathscr{X}_k\}|$. This means that the rank of $S_{n+1}$ \emph{within} $\mathscr{X}_k$ is uniformly distributed on the set $\{1,\ldots,m_k\}$, and so
    \begin{equation}
    \mathbb{P}\{\mathcal{E}^c_{\mathrm{inlier}}\} \leq \frac{2}{m_k} \leq \frac{2}{m}\,.
    \end{equation}

    To upper bound the probability of $\mathcal{E}^c_{\mathrm{cand}}$,  we consider an image similar to pigeonholes. Let the size of the partition be $K = |\wh{\mathfrak{X}}(\mathcal{D}_{1:n+1}) |$. What we know is that, in $\wh{\mathfrak{X}}(\mathcal{D}_{1:n+1})$, each box should contain at least $m$ points. Therefore, we consider randomly assigning $n$ data points to the $K$ holes, while the $n+1$-th points will always go to the $k$-th hole. The event $\mathcal{E}^c_{\mathrm{cand}}$ is then the event that the $k$-th hole receives only $m-1$ points from the $n$ data points. We consider the following steps to scatter. Step 1) Choose $m*K-1$ balls and distribute them so that each hole should have exactly $m$ balls (counting the $n+1$-th point together). There are so many possible ways of doing it, and let's denote each possible event as $A_i$ for $i = 1,....,C$ for some $C$. Step 2) Distribute all remaining $n-Km+1$ balls except in the $k$ th hole, so that this hole remains exactly the size of $m$. Denote this event by $\mathcal{E}_k$. Then, 
  
\[
\P(\mathcal{E}^c_{\mathrm{cand}}) = \sum_{i=1}^C \P(A_i) \P(\mathcal{E}_k|A_i) = \sum_{i=1}^C \P(A_i) \left(1-\frac{1}{K}\right)^{n-Km+1} = \left(1-\frac{1}{K}\right)^{n-Km+1}.
\]

By using \(1-x \leq e^{-x}\) and \(K \leq K_{\rm max}\), we have

\[
\left(1-\frac{1}{K}\right)^{n-Km+1} \leq e^{-\frac{n-Km+1}{K}} \leq \exp\left\{-\left(\frac{n+1}{K_{\rm max}}-m\right)\right\}.
\]

\end{proof}

\subsection{Proof of Theorem \ref{thm:conditional_gu}}\label{sec:proof_conditional_gu}
Define $\mathcal{E} = \mathcal{E}_{\mathrm{inlier}}\cap \mathcal{E}_{\mathrm{cand}}$ with $\mathcal{E}_{\mathrm{inlier}}$ and $\mathcal{E}_{\mathrm{cand}}$ as in \eqref{eq:inliner} and \eqref{eq:cand_set}. {We now  exploit the event $\mathcal{E}$ to characterize the distribution of $S_{n+1}$. Typical split-conformal inference relies on the rank of $S_{n+1}$ being uniformly distributed between $1$ and $n+1$. Lemma \ref{lem:conditional_trick}  leverages the fact that the rank of $S_{n+1}$ {within $\mathscr{X}_k$ is instead uniformly distributed between $2$ and $|\mathscr{X}_k|-1$} to control the probability that $Y_{n+1}$ lies outside of Conformal Tree's prediction set. Denoting $\P_\mE$   the probability measure conditioned on $\mathcal{E}$, this property is detailed in the following lemma.}


\begin{lemma}\label{lem:conditional_trick} 

Let $(X_1,Y_1),\ldots,(X_{n+1},Y_{n+1})$ be {i.i.d. samples from a joint probability measure $\mathbb{P}$}. Let $\wh{\mathfrak{X}}(\mathcal{D}_{1:j})$ denote the partition resulting from Algorithm \ref{alg:robust_tree} on data $ (X_1,S(X_1,Y_1)),\ldots,$ $(X_{j},S(X_{j},Y_{j})))$ for $j\in\{n,n+1\}$. Denote by $m$ the hyperparameter limiting the minimum number of samples in each leaf node. Consider the prediction set  in \eqref{eq:conditional_band} with $\mathfrak{X}=\wh{\mathfrak{X}}(\mathcal{D}_{1:{n}})=\{\mathscr{X}_k\}_{k=1}^K$ and denote it by $\wh{C}^{\wh{\mathfrak{X}}(\mathcal{D}_{1:{n}})}_{n,m} (X_{n+1})$. Define $\mathcal{E} = \mathcal{E}_{\mathrm{inlier}}\cap \mathcal{E}_{\mathrm{cand}}$ with $\mathcal{E}_{\mathrm{inlier}}$ and $\mathcal{E}_{\mathrm{cand}}$ as in \eqref{eq:inliner} and \eqref{eq:cand_set}. Then we have for every $k=1,\dots,K$
\begin{equation}\label{eq:binized_conformal2}
    \P_{\mathcal{E}}\{Y_{n+1}\in \wh{C}^{ \wh{\mathfrak{X}}(\mathcal{D}_{1:{n}})}_{n,m}(X_{n+1})| X_{n+1}\in \mathscr{X}_k, {\wh{\mathfrak{X}}(\mathcal{D}_{1:{n}})}\}\geq 1-\alpha.
\end{equation}
Additionally, provided that the conformity scores
have a continuous joint distribution, we have an upper bound 
 \[\P_{\mathcal{E}}\{Y_{n+1}\in \wh{C}^{\wh{\mathfrak{X}}(\mathcal{D}_{1:{n}})}_{n,m}(X_{n+1})| X_{n+1}\in \mathscr{X}_k, {\wh{\mathfrak{X}}(\mathcal{D}_{1:n})}\}\leq 1-\alpha+\frac{1}{m+1}.\]
\end{lemma}

\begin{proof}
We first consider the properties of $\wh{\mathfrak{X}}(\mathcal{D}_{1:n+1})$ before thinking of $\wh{\mathfrak{X}}(\mathcal{D}_{1:n})$. The corresponding prediction interval in \eqref{eq:conditional_band} with $\mathfrak{X}=\wh{\mathfrak{X}}(\mathcal{D}_{1:n+1})$ is denoted by 
$\wh{C}^{\wh{\mathfrak{X}}(\mathcal{D}_{1:{n+1}})}_{n,m} (X_{n+1})$. Note first as a result of Lemma \ref{lem:exch} that $\{S_i\mid \wh{\mathfrak{X}}(\mathcal{D}_{1:n+1}):i\in [n+1]\}$ are exchangeable random variables as the conditioning statistic $\wh{\mathfrak{X}}(\mathcal{D}_{1:n+1})$ is a symmetric function of the data $(X_1,S_1),\ldots, (X_{n+1}, S_{n+1})$. Consider $\mathscr{X}_k\in \wh{\mathfrak{X}}(\mathcal{D}_{1:n+1})$. Applying \eqref{eq:ground_for_exch}, we   also have that $\{S_i\mid \wh{\mathfrak{X}}(\mathcal{D}_{1:n+1}),X_{n+1}\in\mathscr{X}_k:i\in\mathscr{X}_k\}$, the conformity scores of data within the leaf node containing $X_{n+1}$, are conditionally exchangeable. We finally note that after conditioning on $\mathscr{X}_k$, the event $\mathcal{E}_{\mathrm{cand}}$ is symmetric in the indices $[n+1]$. Therefore, the random variables $\{S_i\mid \wh{\mathfrak{X}}(\mathcal{D}_{1:n+1}),X_{n+1}\in\mathscr{X}_k,\mathcal{E}_{\mathrm{cand}}:i\in\mathscr{X}_k\}$ are exchangeable as well. We denote the number of data points in the leaf $k$ as $m_k:=|\{X_i:X_i\in\mathscr{X}_k\}|$. This means that the rank of $S_{n+1}$ \emph{within} $\mathscr{X}_k$ is uniformly distributed on the set $\{1,\ldots,m_k\}$. Explicitly, letting $\wt{R}_{n+1}$ denote the rank of $S_{n+1}$ within $\mathscr{X}_k$,
\[
\mathbb{P}(\wt{R}_{n+1}=j\mid X_{n+1}\in\mathscr{X}_k, \wh{\mathfrak{X}}(\mathcal{D}_{1:n+1}), \mathcal{E}_{\mathrm{cand}}) = 1/m_k
\]
for any $j\in [m_k]$. When we add the additional condition on $\mathcal{E}_{\mathrm{inlier}}$, the rank of $S_{n+1}$ is conditionally uniform on $\{2,\ldots,m_k-1\}$, as $\mathcal{E}_{\mathrm{inlier}}$ is exactly the event that the rank of $S_{n+1}$ is not $1$ or $m_k$.  This leads to the rank probability 
\[
\mathbb{P}(\wt{R}_{n+1}=j\mid X_{n+1}\in\mathscr{X}_k, \wh{\mathfrak{X}}(\mathcal{D}_{1:n+1}), \mathcal{E}_{\mathrm{inlier}}, \mathcal{E}_{\mathrm{cand}}) = 1/(m_k-2)
\]
for any $j\in \{2,\ldots,m_k-1\}$. Therefore, 
\begin{align}
    \mathbb{P}(Y_{n+1}\not\in \,&\wh{C}^{\wh{\mathfrak{X}}(\mathcal{D}_{1:{ n+1}})}_{n,m}(X_{n+1}) \mid X_{n+1}\in \mathscr{X}_k, \wh{\mathfrak{X}}(\mathcal{D}_{1:n+1}), \mathcal{E}_{\mathrm{inlier}}, \mathcal{E}_{\mathrm{cand}})  \nonumber\\&\leq \mathbb{P}(\wt{R}_{n+1} >\lceil(1-\alpha)(m_k-2)+1\rceil\mid X_{n+1}\in\mathscr{X}_k,\wh{\mathfrak{X}}(\mathcal{D}_{1:n+1}),\mathcal{E}_{\mathrm{inlier}}, \mathcal{E}_{\mathrm{cand}}) \nonumber \\
    &= \frac{(m_k-2)-\lceil (m_k-2)(1-\alpha)+1\rceil+1}{m_k-2} \nonumber \\
    &\leq \frac{(m_k-2)- (m_k-2)(1-\alpha)}{m_k-2} \nonumber\\
    &\leq \alpha\,.\label{eq:mix_match_upper}
\end{align}

If the conformity scores have a conditional joint distribution, we   also have
\begin{align}
    \mathbb{P}(Y_{n+1}\in \,&\wh{C}^{\wh{\mathfrak{X}}(\mathcal{D}_{1:{ n+1}})}_{n,m}(X_{n+1})\mid X_{n+1}\in \mathscr{X}_k, \wh{\mathfrak{X}}(\mathcal{D}_{1:n+1}), \mathcal{E}_{\mathrm{inlier}}, \mathcal{E}_{\mathrm{cand}})  \nonumber\\&= \mathbb{P}(\wt{R}_{n+1}\leq \lceil(1-\alpha)(m_k-2)+1\rceil\mid X_{n+1}\in\mathscr{X}_k,\wh{\mathfrak{X}}(\mathcal{D}_{1:n+1}),\mathcal{E}_{\mathrm{inlier}}, \mathcal{E}_{\mathrm{cand}})  \nonumber\\
    &= \frac{(m_k-2)-\lceil (m_k-2)(1-\alpha)\rceil}{m_k-2}  \nonumber\\
    &\leq 1-\alpha+\frac{1}{m_k-2}  \nonumber\\
    &\leq 1-\alpha + \frac{1}{m-2}\,.\label{eq:mix_match_lower}
\end{align}
Now, note that by \eqref{eq:tree_implication}, we have $\wh{\mathfrak{X}}(\mathcal{D}_{1:n}) =\wh{\mathfrak{X}}(\mathcal{D}_{1:n+1})$ on $ \mathcal{E} = \mathcal{E}_{\mathrm{inlier}}\cap \mathcal{E}_{\mathrm{cand}}.$ Therefore, the bounds in \eqref{eq:mix_match_upper} and \eqref{eq:mix_match_lower} are also bounds for  $ \mathbb{P}(Y_{n+1}\in \wh{C}^{ \wh{\mathfrak{X}}(\mathcal{D}_{1:{ n}})}_{n.m}(X_{n+1})\mid X_{n+1}\in \mathscr{X}_k, \wh{\mathfrak{X}}(\mathcal{D}_{1: n}), \mathcal{E}_{\mathrm{inlier}}, \mathcal{E}_{\mathrm{cand}})$, by which we conclude the proof.

\end{proof}

\paragraph{Proof of Theorem \ref{thm:conditional_gu}}
The proof is by combining Lemma \ref{lem:tree_same} with Lemma \ref{lem:conditional_trick}, by which we can still nearly enjoy the prescribed level of conditional coverage even if we are not conditioning on $\mE$. Note that the notation $\wh{C}_{n,m}(X_{n+1})$ in Theorem \ref{thm:conditional_gu} is equivalent to $\wh{C}^{ \wh{\mathfrak{X}}(\mathcal{D}_{1:{ n}})}_{n.m}(X_{n+1})$ in Lemma \ref{lem:conditional_trick}.

\begin{proof}  
We have
\begin{align*}
    & \mathbb{P}(Y_{n+1}\not\in \wh{C}^{\wh{\mathfrak{X}}(\mathcal{D}_{1:{n}})}_{n,m} (X_{n+1})\mid X_{n+1}\in \mathscr{X}_k, \wh{\mathfrak{X}}(\mathcal{D}_{1:n})) \\
    =\,\,& 
    \mathbb{P}(Y_{n+1}\not\in \wh{C}^{\wh{\mathfrak{X}}(\mathcal{D}_{1:{n}})}_{n,m} (X_{n+1})\mid X_{n+1}\in \mathscr{X}_k, \wh{\mathfrak{X}}(\mathcal{D}_{1:n}), \mathcal{E}) \mathbb{P}(\mathcal{E}\mid X_{n+1}\in \mathscr{X}_k, \wh{\mathfrak{X}}(\mathcal{D}_{1:n})) \\
    &\quad + \mathbb{P}(Y_{n+1}\not\in \wh{C}^{\wh{\mathfrak{X}}(\mathcal{D}_{1:{n}})}_{n,m} (X_{n+1})\mid X_{n+1}\in \mathscr{X}_k, \wh{\mathfrak{X}}(\mathcal{D}_{1:n}), \mathcal{E}^c) \mathbb{P}(\mathcal{E}^c\mid X_{n+1}\in \mathscr{X}_k, \wh{\mathfrak{X}}(\mathcal{D}_{1:n})) \\
    \leq\,\, &  \alpha + \delta(n,m)\,.
\end{align*}
where the inequality follows from Lemma \ref{lem:conditional_trick} and Lemma \ref{lem:bad_bound}\,.

When the conformity scores have a continuous joint distribution, we analogously have
\begin{align*}
    & \mathbb{P}(Y_{n+1}\in \wh{C}^{\wh{\mathfrak{X}}(\mathcal{D}_{1:{n}})}_{n,m} (X_{n+1})\mid X_{n+1}\in \mathscr{X}_k, \wh{\mathfrak{X}}(\mathcal{D}_{1:n})) \\
    =\,\,& 
    \mathbb{P}(Y_{n+1}\in \wh{C}^{\wh{\mathfrak{X}}(\mathcal{D}_{1:{n}})}_{n,m} (X_{n+1})\mid X_{n+1}\in \mathscr{X}_k, \wh{\mathfrak{X}}(\mathcal{D}_{1:n}), \mathcal{E}) \mathbb{P}(\mathcal{E}\mid X_{n+1}\in\mathscr{X}_k,\wh{\mathfrak{X}}(\mathcal{D}_{1:n})) \\
    &\quad + \mathbb{P}(Y_{n+1}\in \wh{C}^{\wh{\mathfrak{X}}(\mathcal{D}_{1:{n}})}_{n,m} (X_{n+1})\mid X_{n+1}\in \mathscr{X}_k, \wh{\mathfrak{X}}(\mathcal{D}_{1:n}), \mathcal{E}^c) \mathbb{P}(\mathcal{E}^c\mid X_{n+1}\in\mathscr{X}_k,\wh{\mathfrak{X}}(\mathcal{D}_{1:n})) \\
    \leq\,\, &  1-\alpha + \frac{1}{m-2}+\delta(n,m)\,.
\end{align*}
where the inequality again uses the upper bound in Lemma \ref{lem:conditional_trick} and Lemma \ref{lem:bad_bound}.

\end{proof}

\subsection{Proof of Theorem \ref{thm:main}}\label{pr:main_thm}

\begin{proof}

Denote by $\mathbb{X}$, the set of all the partitions associated with the tree class in use. For a partition $\wt{\mathfrak{X}}\in \mathbb{X}$, define 
\begin{equation}\label{eq:binized_conformal}
    P_{n+1}^{k,\wt{\mathfrak{X}}}:=\P\{Y_{n+1}\in \wh{C}^{ \wh{\mathfrak{X}}(\mathcal{D}_{1:n})}_{n,m}(X_{n+1})| X_{n+1}\in \mathscr{X}_k, \wh{\mathfrak{X}}(\mathcal{D}_{1:n})=\wt{\mathfrak{X}}\}.
\end{equation}
    By Theorem \ref{thm:conditional_gu}, we have $P_{n+1}^{k,\wt{\mathfrak{X}}}\geq 1-\alpha$. Then we have 
    \begin{align*}
    \mathbb{P}\{Y_{n+1}\in \wh{C}_{n,m}(X_{n+1})\} &= \sum_{\wt{\mathfrak{X}}\in \mathbb{X}} \sum_{\mX_k\in \wt{\mathfrak{X}}} P_{n+1}^{k,\wt{\mathfrak{X}}}
    \P(X_{n+1}\in \mathscr{X}_k|\wh{\mathfrak{X}}(\mathcal{D}_{1:n})=\wt{\mathfrak{X}})
    \P(\wh{\mathfrak{X}}(\mathcal{D}_{1:n})=\wt{\mathfrak{X}})\\
    &\geq (1-\alpha + \delta(n,m)) \sum_{\wt{\mathfrak{X}}\in \mathbb{X}} \sum_{\mX_k\in \wt{\mathfrak{X}}}\P(X_{n+1}\in \mathscr{X}_k|\wh{\mathfrak{X}}(\mathcal{D}_{1:n})=\wt{\mathfrak{X}})
    \\&~~~~~~~~~~~~~~~~~~~~~~~~~~~~~~~~~~~~~~\cdot\P(\wh{\mathfrak{X}}(\mathcal{D}_{1:n})=\wt{\mathfrak{X}})\\
    &= 1-\alpha + \delta(n,m)\,.
\end{align*}
When the conformity scores have a continuous joint distribution, the upper bound is shown analogously using the upper bound in Theorem \ref{thm:conditional_gu}.

\end{proof}

\section{Conformal Prediction through Robust Tree Refitting}\label{sec:tighter}
\subsection{Hallucination Algorithm}

 Here, we also consider a variation from Algorithm \ref{alg:robust_tree} to take $X_{n+1}$ into account without using the true $Y_{n+1}$ value. It leads to a tighter coverage bound, where the price is the computational cost for refitting the tree for every $X_{n+1}$ value.

\begin{algorithm}[!h]
\caption{Conformal Tree (hallucinating $S_{n+1}$)}\label{alg:conf_tree_mod}
\spacingset{1.2}
\begin{algorithmic}[1]
\Require Fitted model $\hat{\mu}$, calibration dataset $\mathcal{D}_{1:n}$, tree hyperparameters ($m,{K_{\rm max}},\eta=0.05$), test point $X_{n+1}$
\State Compute conformity scores $S_i=S(X_i,Y_i)$ for $i\in [n]$ 
\State Fit a tree $\hat \mT(\mathcal D_{1:n})$ with Algorithm \ref{alg:robust_tree} on $\{(X_i,S(X_i,Y_i)):i\in[n]\}\cup\{(X_{n+1},\wt{S}_{n+1})\}$, where $\wt{S}_{n+1}$ is always considered to be the center of the range of conformity scores in the current leaf node containing $X_{n+1}$
\State Define a partition $\wh{\mathfrak{X}}(\mathcal{D}_{1:n})=\{\mathscr{X}_1,\ldots,\mathscr{X}_K\}$ induced by leaves of $\hat \mT(\mathcal D_{1:n}).$\label{ln:rule1}
\State For $k=1,...,K$, compute $S^*(\mathcal{D}_{1:n}^{k,\wh{\mathfrak{X}}(\mathcal{D}_{1:n})})$, the $\lceil (1-\alpha)\cdot (m_k-2)+1\rceil/m_k$-th smallest value of $\{S(X_i,Y_i): i \in[n],X_i\in\mathscr{X}_k\}$, where $m_k$ is the cardinality of this set.\label{ln:rule2}
\State Construct a prediction set $\wh{C}'_{n,m}$ as 
\[
\wh{C}'_{n,m}(x) = \{y:S(X_{n+1},y)\leq S^*(x)\},
\]
where $S^*(x)=S^*(\mathcal{D}_{1:n}^{k,\wh{\mathfrak{X}}(\mathcal{D}_{1:n})}):x\in\mathscr{X}_{k}$\,.
\end{algorithmic}
\end{algorithm}
 
 {Algorithm \ref{alg:conf_tree_mod} differs from Algorithm \ref{alg:robust_tree} in that it uses a value for $X_{n+1}$ by \emph{hallucinating} a conformity score $\wt{S}_{n+1}$ that will never affect the criterion in any node in which $X_{n+1}$ lies. For example, $\wt{S}_{n+1}$ can be reset prior to each splitting decision to be equal to any other $S_i$ lying inside the same node as $X_{n+1}$. Because of this, $\wt{S}_{n+1}$ will never affect the criterion value in any node, and the tree will be fit as if the data point $(X_{n+1},\wt{S}_{n+1})$ were excluded from any criterion calculations. Therefore, this algorithm avoids any possibility of $(X_{n+1},S_{n+1})$ changing the tree by changing the candidacy status of a potential leaf node. The computational cost is equivalent when computing a prediction set for a single value of $X_{n+1}$, but scales linearly in the number of test points at which to evaluate $\wh{C}'_{n,m}(\cdot)$, whereas Algorithm \ref{alg:conf_tree} has constant cost as a function of the number of test points at which to evaluate $\wh{C}_{n,m}(\cdot)$.}
 
 \subsection{Theoretical Coverage of Hallucination Algorithm}
Now, we present that with the refitting algorithm in Algorithm \ref{alg:conf_tree_mod}, we have improved the bounds due to the additional knowledge on $X_{n+1}$. In this section, we use the following notations to denote partitions: Let $\wh{\mathfrak{X}}_n(X_{n+1})=\wh{\mathfrak{X}}(\mathcal{D}_{1:n},X_{n+1})$ denote the partition resulting from Algorithm \ref{alg:conf_tree_mod} on data $(X_1,Y_1),\ldots,(X_n,Y_n)$ and $X_{n+1}$. Let $\wh{\mathfrak{X}}_{n+1}=\wh{\mathfrak{X}}(\mathcal{D}_{1:n+1})$ denote the partition resulting from Algorithm \ref{alg:robust_tree} on $(X_1,Y_1),\ldots,(X_{n+1},Y_{n+1})$.

\begin{thm} \label{thm:new} 
Under the same setting as in Theorem \ref{thm:conditional_gu}, but with the tree fitting algorithm as Algorithm \ref{alg:conf_tree_mod} instead of Algorithm \ref{alg:robust_tree}, denote the fitted partition as $\wh{\mathfrak{X}}(\mathcal{D}_{1:n},X_{n+1})  = \{\mathscr{X}_k\}_{k=1}^K.$ Then with the prediction set $\wh{C}'_{n,m}(X_{n+1})$ defined in Algorithm \ref{alg:conf_tree_mod}, we have
    \[
    \mathbb{P}\{Y_{n+1}\in \wh{C}'_{n,m}(X_{n+1}) \mid X_{n+1} \in \mathscr{X}_k , \wh{\mathfrak{X}}(\mathcal{D}_{1:n},X_{n+1})\} \geq  1-\alpha-\frac{2}{m}\,.
    \]
    Furthermore, if the conformity scores have a continuous joint distribution, we have
    \[
    \mathbb{P}\{Y_{n+1}\in\wh{C}'_{n,m}(X_{n+1}) \mid X_{n+1} \in \mathscr{X}_k , \wh{\mathfrak{X}}(\mathcal{D}_{1:n},X_{n+1})\} \leq  1-\alpha+\frac{1}{m-2}+\frac{2}{m}\,.
    \]
\end{thm}
\begin{proof}
    We have
\begin{align*}
    & \mathbb{P}(Y_{n+1}\not\in \wh{C}_{n,m}^{\wh{\mathfrak{X}}_{n}(X_{n+1})}(X_{n+1})\mid X_{n+1}\in \mathscr{X}_k, \,\wh{\mathfrak{X}}_n(X_{n+1})) \\
    =\,\,& 
    \mathbb{P}(Y_{n+1}\not\in \wh{C}_{n,m}^{\wh{\mathfrak{X}}_{n}(X_{n+1})}(X_{n+1})\mid X_{n+1}\in \mathscr{X}_k, \,\wh{\mathfrak{X}}_n(X_{n+1}), \mathcal{E}_{\mathrm{inlier}}) \mathbb{P}(\mathcal{E}_{\mathrm{inlier}}\mid X_{n+1}\in \mathscr{X}_k, \,\wh{\mathfrak{X}}_n(X_{n+1})) \\
    &\quad + \mathbb{P}(Y_{n+1}\not\in \wh{C}_{n,m}^{\wh{\mathfrak{X}}_{n}(X_{n+1})}(X_{n+1})\mid X_{n+1}\in \mathscr{X}_k, \,\wh{\mathfrak{X}}_n(X_{n+1}), \mathcal{E}^c_{\mathrm{inlier}}) \mathbb{P}(\mathcal{E}^c_{\mathrm{inlier}}\mid X_{n+1}\in \mathscr{X}_k, \,\wh{\mathfrak{X}}_n(X_{n+1})) \\
    \leq\,\, &  \alpha + \frac{2}{m_k}
    \leq\,\, \alpha + \frac{2}{m}\,,
\end{align*}
where we applied Lemma \ref{lem:new2} and used $ \mathbb{P}(\mathcal{E}^c_{\mathrm{inlier}}\mid X_{n+1}\in \mathscr{X}_k, \,\wh{\mathfrak{X}}_n(X_{n+1})) \leq 2/m_k$ due to the uniformity of the rank of the scores\,.

 Now, we consider the upper bound. Note that on $\mathcal{E}_{\mathrm{inlier}}$ we have $\wh{\mathfrak{X}}_{n+1}=\wh{\mathfrak{X}}_n(X_{n+1}) $ due to Lemma \ref{lem:inlier_same}. Therefore, when the conformity scores have a continuous joint distribution, we have
\begin{align*}
    & \mathbb{P}(Y_{n+1}\in \wh{C}_{n,m}^{\wh{\mathfrak{X}}_n(X_{n+1})}(X_{n+1})\mid X_{n+1}\in \mathscr{X}_k, \wh{\mathfrak{X}}_n(X_{n+1})) \\
    =\,\,& 
    \mathbb{P}(Y_{n+1}\in \wh{C}_{n,m}^{\wh{\mathfrak{X}}_{n+1}}(X_{n+1})\mid X_{n+1}\in \mathscr{X}_k, \wh{\mathfrak{X}}_{n+1}, \mathcal{E}_{\mathrm{inlier}}) \mathbb{P}(\mathcal{E}_{\mathrm{inlier}}) \\
    &\quad + \mathbb{P}(Y_{n+1}\in \wh{C}_{n,m}^{\wh{\mathfrak{X}}_n(X_{n+1})}(X_{n+1})\mid X_{n+1}\in \mathscr{X}_k, \wh{\mathfrak{X}}_{n}(X_{n+1}), \mathcal{E}^c_{\mathrm{inlier}}) \mathbb{P}(\mathcal{E}^c_{\mathrm{inlier}}) \\
    \leq\,\, &  1-\alpha + \frac{1}{m-2}+\frac{2}{m}\,. 
\end{align*}

Finally, we note that marginal coverage bounds can be obtained in a similar manner to Theorem \ref{thm:main}.
\end{proof}

\begin{lemma}\label{lem:new2}
Let $(X_1,Y_1),\ldots,(X_{n+1},Y_{n+1})$ be i.i.d. samples from a joint probability measure $\mathbb{P}$. Assume we apply Algorithm \ref{alg:conf_tree_mod} on $(X_1,S(X_1,Y_1)),\ldots,(X_n,S(X_n,Y_n))$ for test point $X_{n+1}$ with resulting partition $\wh{\mathfrak{X}}_n(X_{n+1})$. Then with $\wh{C}_{n,m}^{\wh{\mathfrak{X}}_{n}(X_{n+1})}(X_{n+1})$ obtained using Algorithm \ref{alg:conf_tree_mod}, and $\wh{\mathfrak{X}}_{n+1}$ the partition obtained using Algorithm 1 on all $n+1$ data points, we have
\[
    \mathbb{P}(Y_{n+1}\not\in \wh{C}_{n,m}^{\wh{\mathfrak{X}}_{n}(X_{n+1})}(X_{n+1})\mid X_{n+1}\in \mathscr{X}_k, \wh{\mathfrak{X}}_{n}(X_{n+1}), \mathcal{E}_{\mathrm{inlier}}) \leq \alpha
    \]
    where $\mathcal{E}_{\mathrm{inlier}}$ is as defined in \eqref{eq:alg3_inlierset}. Furthermore, if the conformity scores have a continuous joint distribution, then
    \[
    \mathbb{P}(Y_{n+1}\in \wh{C}_{n,m}^{\wh{\mathfrak{X}}_{n}(X_{n+1})}(X_{n+1})\mid X_{n+1}\in \mathscr{X}_k, \wh{\mathfrak{X}}_{n}(X_{n+1}), \mathcal{E}_{\mathrm{inlier}}) \geq 1-\alpha + 1/(m-2)\,.
    \]
\end{lemma}
\begin{proof}
Note first as a result of Lemma \ref{lem:exch} that $\{S_i\mid \wh{\mathfrak{X}}_{n+1}:i\in [n+1]\}$ are exchangeable random variables as the conditioning statistic $\wh{\mathfrak{X}}_{n+1}$ is a symmetric function of the data $(X_1,S_1),\ldots, (X_{n+1}, S_{n+1})$. Applying \eqref{eq:ground_for_exch}, we have also that $\{S_i|\wh{\mathfrak{X}}_{n+1},X_{n+1}\in\mathscr{X}_k:i\in\mathscr{X}_k\}$, the conformity scores of data within the leaf node containing $X_{n+1}$, are conditionally exchangeable. We denote the number of data points in leaf $k$ as $m_k:=|\{X_i:X_i\in\mathscr{X}_k\}|$. This means that the rank of $S_{n+1}$ \emph{within} $\mathscr{X}_k$ is uniformly distributed on the set $\{1,\ldots,m_k\}$. Explicitly, letting $\wt{R}_{n+1}$ denote the rank of $S_{n+1}$ within $\mathscr{X}_k$,
\[
\mathbb{P}(\wt{R}_{n+1}=j\mid X_{n+1}\in\mathscr{X}_k, \wh{\mathfrak{X}}_{n+1}) = 1/m_k
\]
for any $j\in [m_k]$. When we add the additional condition on $\mathcal{E}_{\mathrm{inlier}}$, the rank of $S_{n+1}$ is conditionally uniform on $\{2,\ldots,m_k-1\}$, as $\mathcal{E}_{\mathrm{inlier}}$ is exactly the event that the rank of $S_{n+1}$ is not $1$ or $m_k$. Thus 
\[
\mathbb{P}(\wt{R}_{n+1}=j\mid X_{n+1}\in\mathscr{X}_k, \wh{\mathfrak{X}}_{n+1}, \mathcal{E}_{\mathrm{inlier}}) = 1/(m_k-2)
\]
for any $j\in \{2,\ldots,m_k-1\}$. Therefore, 
\begin{align*}
    \mathbb{P}(Y_{n+1}\not\in \,&\wh{C}_{n,m}^{\wh{\mathfrak{X}}_{n+1}}(X_{n+1})\mid X_{n+1}\in \mathscr{X}_k, \wh{\mathfrak{X}}_{n+1}, \mathcal{E}_{\mathrm{inlier}}) \\&= \mathbb{P}(\wt{R}_{n+1} > \lceil(1-\alpha)(m_k-2)+1\rceil\mid X_{n+1}\in\mathscr{X}_k,\wh{\mathfrak{X}}_{n+1},\mathcal{E}_{\mathrm{inlier}}) \\
    &\leq \frac{(m_k-2)-\lceil (m_k-2)(1-\alpha)+1\rceil+1}{m_k-2} \\
    &\leq \frac{(m_k-2)- (m_k-2)(1-\alpha)}{m_k-2} \\
    &\leq \alpha\,.
\end{align*}
Furthermore, if the conformity scores have a continuous joint distribution, then we have 
\begin{align*}
    \mathbb{P}(Y_{n+1}\in\, &\wh{C}_{n,m}^{\wh{\mathfrak{X}}_{n+1}}(X_{n+1})\mid X_{n+1}\in \mathscr{X}_k, \wh{\mathfrak{X}}_{n+1}, \mathcal{E}_{\mathrm{inlier}}) \\ 
    &= \mathbb{P}(\wt{R}_{n+1}\leq \lceil(1-\alpha)(m_k-2)+1\rceil\mid X_{n+1}\in\mathscr{X}_k,\wh{\mathfrak{X}}_{n+1},\mathcal{E}_{\mathrm{inlier}}) \\
    &= \frac{\lceil (m_k-2)(1-\alpha)+1\rceil-1}{m_k-2} \\
    &= \frac{\lceil (m_k-2)(1-\alpha)\rceil}{m_k-2} \\
    &\leq 1-\alpha + \frac{1}{m_k-2} \\
    &\leq 1-\alpha + \frac{1}{m-2}\,.
\end{align*}
Now, by Lemma \ref{lem:inlier_same}, noting that on $\mathcal{E}_{\mathrm{inlier}}$ we have $\wh{\mathfrak{X}}_{n+1}=\wh{\mathfrak{X}}_n(X_{n+1}) $, we conclude the proof.
\end{proof}

\begin{lemma}\label{lem:inlier_same}
    Let $(X_1,Y_1),\ldots,(X_{n+1},Y_{n+1})$ be i.i.d. samples from a joint probability measure $\mathbb{P}$. Let $\mathscr{X}_{k(X_{n+1})}$ be the leaf of $\wh{\mathfrak{X}}_n(X_{n+1})$ containing $X_{n+1}$.
    Define 
    \begin{equation}\label{eq:alg3_inlierset}
        \mathcal{E}_{\mathrm{inlier}}=\{\min_{\substack{i:i \in [n] \\ X_i \in \mathscr{X}_{k(X_{n+1})}}} S_i \leq S_{n+1} \leq \max_{\substack{i:i \in [n] \\ X_i \in \mathscr{X}_{k(X_{n+1})}}} S_i\}\,,
    \end{equation}
    meaning $\mathcal{E}_{\mathrm{inlier}}$ is the event on which $S_{n+1}$ does not have the uniquely largest or smallest conformity score among those in its leaf node. Then under event $\mathcal{E}_{\mathrm{inlier}}$, $\wh{\mathfrak{X}}_n(X_{n+1}) = \wh{\mathfrak{X}}_{n+1}$ with probability 1.
    \label{lem:new1_new_alg}
\end{lemma}

\begin{proof}
    Note that the only possibility for $\wh{\mathfrak{X}}_{n+1}$ to differ from $\wh{\mathfrak{X}}_n(X_{n+1})$ requires $S_{n+1}$ to be the unique maximizer or minimizer inside a tree node. Since leaf nodes are subsets of any intermediate nodes, it also requires $S_{n+1}$ to be the unique maximizer or minimizer of the leaf node in which $X_{n+1}$ resides. Otherwise, it will never have had an impact in the tree fitting process. Thus, $\wh{\mathfrak{X}}_n(X_{n+1}) \neq \wh{\mathfrak{X}}_{n+1}$ implies $\mathcal{E}^c_{\mathrm{inlier}}$. 
\end{proof}

\subsection{Empirical Consideration of Hallucination Algorithm}

\begin{table}[!t]
\centering
\begin{minipage}{\textwidth}
\resizebox{\textwidth}{!}{%
\begin{tabular}{|l|cc|ccc|ccc|ccc|ccc|ccc|ccc|}
\hline
 & \multicolumn{2}{c|}{Split Conformal} & \multicolumn{3}{c|}{Conf. Tree (Hallucination)} & \multicolumn{3}{c|}{Conformal Tree} & \multicolumn{3}{c|}{Conf. Tree (CART)} & \multicolumn{3}{c|}{Conformal Forest} & \multicolumn{3}{c|}{Locally Weighted)} & \multicolumn{3}{c|}{CQR} \\
 & Width & Cov. & Width & Cov. & P.B. & Width & Cov. & P.B. & Width & Cov. & P.B. & Width & Cov. & P.B. & Width & Cov. & P.B. & Width & Cov. & P.B. \\
\hline
Data 1 & 4.06 & 0.9 & \textbf{3.24} & 0.86 & \textbf{0.66} & \textbf{3.24} & 0.86 & \textbf{0.66} & 4.04 & 0.93 & 0.54 & 3.85 & 0.9 & 0.54 & 4.31 & 0.91 & 0.55 & 3.82 & 0.89 & 0.64 \\
Data 2 & 2.48 & 0.87 & 3.1 & 0.82 & \textbf{0.81} & 3.1 & 0.82 & \textbf{0.81} & 4.1 & 0.9 & 0.61 & 3.65 & 0.91 & 0.36 & 4.01 & 0.87 & 0.36 & 3.11 & 0.88 & 0.3 \\
bike & 2.05 & 0.93 & 2.0 & 0.9 & \textbf{0.67} & 2.06 & 0.91 & 0.66 & 2.4 & 0.95 & 0.43 & 2.12 & 0.93 & 0.49 & 2.17 & 0.95 & 0.46 & \textbf{1.9} & 0.91 & 0.57 \\
bio & 2.19 & 0.87 & 2.05 & 0.85 & 0.65 & 2.04 & 0.85 & \textbf{0.68} & 2.36 & 0.9 & 0.36 & 2.61 & 0.9 & 0.33 & 2.24 & 0.87 & 0.48 & \textbf{1.9} & 0.89 & 0.62 \\
community & 2.15 & 0.92 & 1.87 & 0.89 & 0.67 & 1.89 & 0.9 & 0.69 & 2.4 & 0.96 & 0.56 & 2.56 & 0.94 & 0.48 & 2.06 & 0.95 & 0.6 & \textbf{1.84} & 0.91 & \textbf{0.72} \\
concrete & 0.8 & 0.91 & \textbf{0.73} & 0.88 & 0.62 & \textbf{0.73} & 0.87 & \textbf{0.65} & 0.81 & 0.9 & 0.46 & 0.83 & 0.89 & 0.51 & 0.77 & 0.88 & 0.6 & 0.88 & 0.88 & 0.4 \\
homes & 0.9 & 0.9 & 0.93 & 0.88 & \textbf{0.79} & 0.96 & 0.88 & 0.77 & 1.28 & 0.9 & 0.64 & 1.17 & 0.9 & 0.46 & \textbf{0.84} & 0.9 & 0.69 & 0.93 & 0.91 & 0.53 \\
star & 0.18 & 0.9 & \textbf{0.16} & 0.85 & \textbf{0.7} & \textbf{0.16} & 0.86 & \textbf{0.7} & 0.2 & 0.93 & 0.24 & 0.2 & 0.91 & 0.44 & 0.2 & 0.92 & 0.35 & 0.2 & 0.92 & 0.28 \\
\hline
\end{tabular}
}
\end{minipage}
\caption{Comparison of width, coverage, and prediction bands (P.B.) across different methods and datasets, including the hallucination algorithm. The experiments were conducted on the randomly selected 500 data points in order to run the Hallucination algorithm in a reasonable amount of time.}
\label{table:results-hallucination}
\end{table}

We repeat the experiment described earlier in Section \ref{sec:benchmark} using the hallucination version of the algorithm which boasts a tighter coverage guarantee. However, this algorithm is significantly more computationally expensive, as it requires the robust tree to be fit separately for each test point, rather than just once in total. For this reason, in this section, we limit the size of each dataset to 500 randomly chosen observations. The results are in Table \ref{table:results-hallucination}.

{We find the performance of Conformal Tree with and without hallucination to be very similar. This is because it will be rare that the inclusion of the covariate value for the test point adds a new eligible split \emph{and} that this split is chosen among the split candidates. Both of these conditions are necessary for the tree from the hallucination algorithm to differ from the non-hallucination version. As this is fairly unlikely, we expect the two methods to perform similarly in most cases. While some of the intervals end up slightly sharper using the hallucination algorithm, we generally recommend against it if a low runtime is of concern to the practitioner. This is because the hallucination algorithm requires a new robust tree to be fit for each test point, which can add significant cost if the {test set} is large. On the other hand, if the user were only interested in a prediction set for a single test point, there would be almost no downside to the hallucination algorithm.}

\section{Sensitivity Analysis for $K_{\rm max}$ and $m$} \label{sec:hyp_sense}
For each dataset, we vary $m$ and $K_{\rm max}$ over a grid of values in order to determine the relative effect of varying these hyperparameters on the performance of our method. We plot the resulting average interval score length on held-out data for each combination of $m$ and $K_{\rm max}$ in Figure~\ref{fig:k-m-analysis1}. The figure shows that the best results are generally obtained with relatively small $m$, and often with the large $K_{\rm max}$ as possible, implying that there is a significant benefit in capturing more local structure in the data. 
    With this in mind, we recommend that practitioners choose the smallest $m$ as possible where $m$ is chosen based on their desired worst-case coverage computed in Theorem \ref{thm:conditional_gu}. Practitioners are welcome to tune hyperparameters using additional held-out data if they so desire as well. 

\begin{figure}
    \centering
    \includegraphics[width=0.9\linewidth]{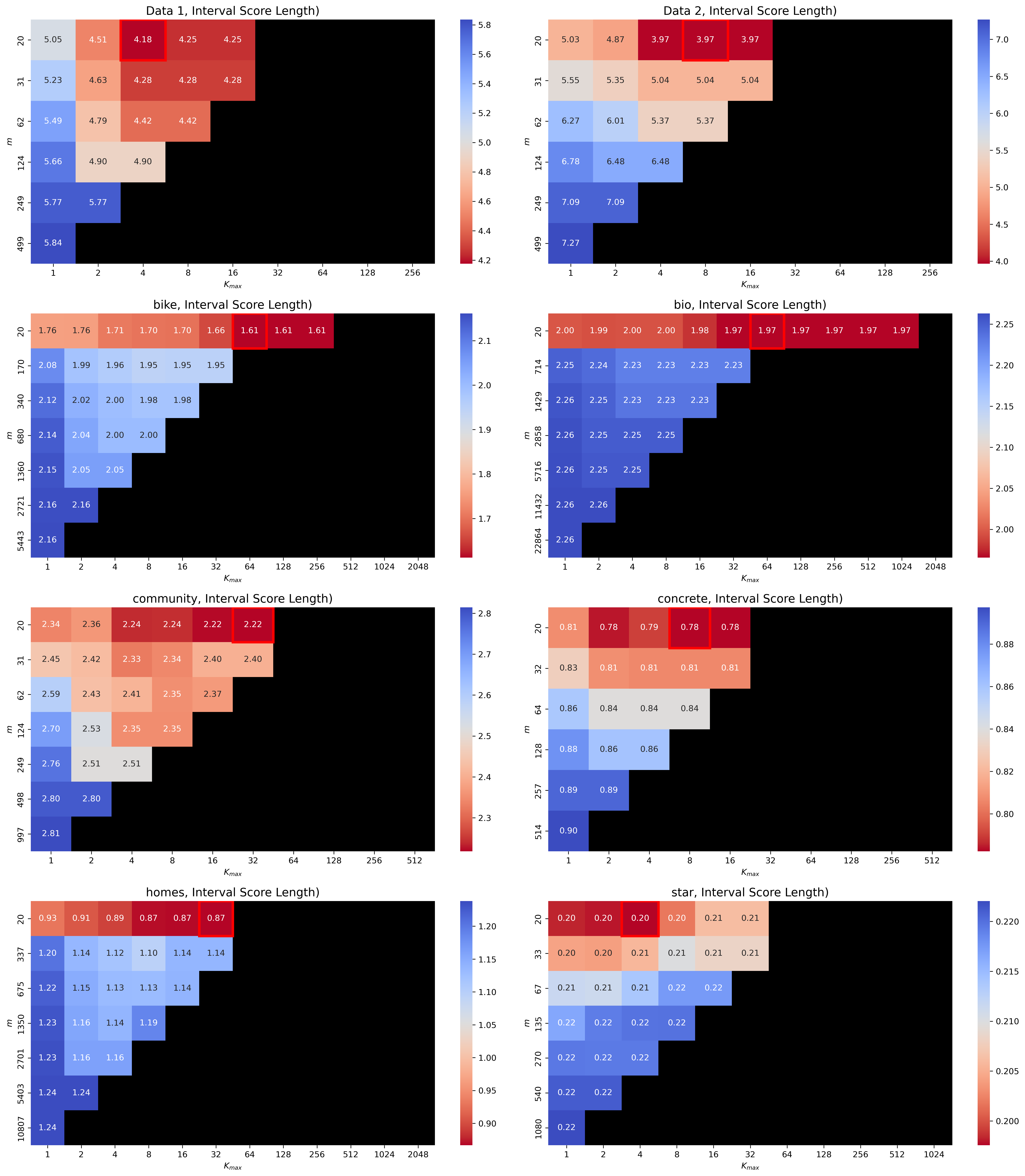}
    \caption{Comparison of varying values for $m$ and ${K_{\rm max}}$ on interval score length on each dataset. Metrics are computed on a held-out validation data split. The red-lined lined box indicates the optimal (achieving the smallest interval score length).}
    \label{fig:k-m-analysis1}
\end{figure}


\section{Experimental Details}
\label{sec:exp-info}

\subsection{Details for Section \ref{sec:benchmark}}\label{sec:full_benchmark_detail}

For the dataset, besides the synthetic datasets (Data 1 and Data 2) in Section \ref{sec:exp_adapt}, we use (\texttt{bike}, \texttt{bio}, \texttt{community}, \texttt{concrete}, \texttt{homes}, \texttt{star}) for regression problems that have been previously used as benchmarks in conformal prediction \citep{rossellini2023integrating}. As the predictive model for all datasets, we used a random forest model with 100 trees, fit using the variance reduction criterion, requiring a minimum of two samples to split. 

For Conformal Tree, we used a fixed minimum improvement threshold of 0.05, and tuned the hyperparameters $m$ and $K_{\mathrm{max}}$ on each dataset in order to {minimize the interval score loss}. Conformal Forest mentioned in Remark \ref{rm:forest} combines intervals from 100 CART trees, with randomness arising from bagging and feature subsampling. The intervals are combined using a recent majority vote procedure referenced in the text \citep{gasparin2024merging}. Each tree in the forest also is regularized using the same hyperparameters as described earlier. The Locally weighted method uses an auxiliary random forest model to predict the absolute residuals. The hyperparameters of this model are tuned analogously, as we select the combination that minimizes interval score length on held-out validation data. For CQR, the quantiles are predicted in the manner of \citet{meinshausen2006quantile}. All methods other than CQR only utilize the predictions of the conditional mean from the model. 

\subsection{Large Language Model Uncertainty Quantification}

In the experiments in Secion \ref{sec:LLM}, we use a large language model for downstream classification tasks. We used the GPT-4o model produced by OpenAI in June of 2024 for our classification problems. We used a temperature setting of 0.0, yielding (ideally) deterministic output, except for our naive uncertainty quantification comparisons, in which we used a temperature setting of 1.0.

\subsubsection{Skin-disease Prediction}\label{sec:full_skin_detail}
The calibration data for the skin disease task is found in the UCI machine learning
repository \citep{misc_dermatology_33}. The dataset consists of 366 observations. We use the twelve clinical features, and ignore the remaining 23 histopathological features,
as the purpose of this experiment is to emulate a user describing physical symptoms to a language model. The prompt used for the classification task is given below.

\begin{lstlisting}
Predict the diagnosis of Eryhemato-Squamous disease in the following case, using the following clinical features. The age feature simply represents the age of
    the patient. Every other feature was given a degree in the range of 0 to 3. Here, 0 indicates that the feature was not present, 3 indicates the largest amount possible, and 1, 2 indicate the relative intermediate values.

    erythema: 2.0, scaling: 2.0, definite_borders: 0.0, itching: 3.0, koebner phenomenon: 0.0, polygonal papules: 0.0, follicular papules: 0.0, oral mucosal involvement: 0.0, knee and elbow involvement: 1.0, scalp involvement: 0.0, family history: 0.0, age: 55.0

    The possible classes are: psoriasis, seboreic dermatitis, lichen planus, pityriasis rosea, cronic dermatitis, pityriasis rubra pilaris.

    Please estimate the probability of each possible diagnosis for this case. The following is for research purposes only. I understand that a real patient must see a qualified doctor with such a concern.

    Format your answer as:
    psoriasis: (prob),
    seboreic dermatitis: (prob),
    ...

    Do your best to provide an accurate answer strictly in this format, and do not include anything else in your response.
\end{lstlisting}

\subsubsection{Classifying States of Legislators}

The task in this experiment is to predict U.S. legislators' state of employment from their DW-NOMINATE scores \citep{poole1985spatial}. These are estimates of two-dimensional latent variables in a spatial model based on roll-call voting. To an unfamiliar reader, they can loosely interpreted as two-dimensional summary statistics pertaining to the legislators' ideology.

The dataset of legislators and estimates of their DW-NOMINATE scores was obtained from Voteview \citep{lewis2024voteview}. We filter the data by only including congresses 101-117. We consider each unique legislator from one of the fifty U.S. states. For duplicate entries of legislators, we include only the first occurrence of that individual. The prompt used for the classification task is given below.

\begin{lstlisting}
Predict the U.S. state of a legislator with the following DW-NOMINATE score: (-0.4, 0.11).
Estimate the probability that they were from each of the 50 U.S States. Format your answer as
    Alabama: 0.01,
    Alaska: 0.02,
    Arizona: 0.03,
    ...

Do your best to provide an accurate answer in this format.
\end{lstlisting}



The naive uncertainty quantification method works by sampling various classification responses from the language model with the temperature set to 1.0. For our experiments, we sampled 11 classification responses (simplices) for each test point. We construct a prediction set for each sample by taking the minimum threshold such that including all labels with probability mass above that threshold causes the total mass to exceed $1-\alpha$. We then combine the $11$ sets for each test point by taking a majority vote. This method produces a naive, uncalibrated prediction set, which we use solely for comparisons.


\section{Additional Experimental Results}\label{sec:add_vis}

In this section, we include additional visualizations referenced from the main body.

\paragraph{Constant calibration set size} We repeat the experiment in Section \ref{sec:benchmark} with a fixed size of $n=500$ points in the calibration set for each approach, rather than a proportion of the total data size. The results are in Table \ref{table:all-tuned-results-forest-small-tree} when a single tree is utilized and in Table \ref{table:all-tuned-results-forest-small-tree} when a forest is used.

\begin{table}[!t]
\centering
\resizebox{\columnwidth}{!}{
\begin{tabular}{|l|ccc|cccc|cccc|cccc|}
\hline
 & \multicolumn{3}{c|}{Split Conformal} & \multicolumn{4}{c|}{Conformal Tree} & \multicolumn{4}{c|}{Locally Weighted (Tree)} & \multicolumn{4}{c|}{CQR (Tree)} \\
 & Width & Cov. & ISL & Width & Cov. & ISL & P.B. & Width & Cov. & ISL & P.B. & Width & Cov. & ISL & P.B. \\
\hline
Data 1 & 4.29 & 0.93 & 5.37 & \textbf{3.59} & 0.91 & \textbf{4.33} & 0.59 & 3.85 & 0.92 & 4.39 & 0.59 & 4.57 & 0.88 & 6.16 & \textbf{0.8} \\
Data 2 & \textbf{2.69} & 0.9 & 5.07 & 2.87 & 0.87 & \textbf{3.85} & \textbf{0.87} & 3.31 & 0.91 & 4.27 & 0.58 & 5.29 & 0.9 & 9.51 & 0.0 \\
bike & 1.18 & 0.91 & 1.83 & \textbf{1.11} & 0.87 & 1.83 & \textbf{0.59} & 1.13 & 0.89 & \textbf{1.61} & \textbf{0.59} & 3.63 & 0.89 & 5.0 & 0.0 \\
bio & 1.53 & 0.89 & \textbf{2.12} & \textbf{1.48} & 0.88 & 2.15 & \textbf{0.48} & 1.62 & 0.91 & 2.14 & 0.29 & 3.41 & 0.91 & 3.61 & 0.0 \\
community & 1.83 & 0.9 & 2.77 & \textbf{1.68} & 0.88 & 2.6 & \textbf{0.67} & 1.82 & 0.91 & \textbf{2.47} & 0.59 & 3.03 & 0.89 & 4.66 & 0.0 \\
concrete & 0.67 & 0.91 & \textbf{0.87} & \textbf{0.63} & 0.88 & 0.88 & \textbf{0.62} & 0.65 & 0.9 & 0.89 & 0.53 & 1.97 & 0.88 & 2.37 & 0.0 \\
homes & \textbf{0.58} & 0.9 & 1.13 & \textbf{0.58} & 0.88 & 1.08 & \textbf{0.62} & 0.63 & 0.88 & \textbf{1.06} & 0.59 & 3.49 & 0.91 & 4.19 & 0.0 \\
star & 0.17 & 0.89 & \textbf{0.21} & \textbf{0.16} & 0.87 & 0.22 & \textbf{0.55} & 0.17 & 0.89 & \textbf{0.21} & 0.4 & 0.24 & 0.89 & 0.29 & 0.0 \\
\hline
\end{tabular}
}
\caption{Empirical results for tree-based methods on real and synthetic data as the average over five trials on width, coverage, ISL, and proportion better (P.B.) for $\alpha=0.1$. Calibration sets all have the size $n=500$.} \label{table:all-tuned-results-forest-small-tree}
\end{table}

\begin{table}[!t]
\centering
\resizebox{\columnwidth}{!}{
\begin{tabular}{|l|ccc|cccc|cccc|cccc|}
\hline
 & \multicolumn{3}{c|}{Split Conformal} & \multicolumn{4}{c|}{Conformal Forest} & \multicolumn{4}{c|}{Locally Weighted (Forest)} & \multicolumn{4}{c|}{CQR (Forest)} \\
 & Width & Cov. & ISL & Width & Cov. & ISL & P.B. & Width & Cov. & ISL & P.B. & Width & Cov. & ISL & P.B. \\
\hline
Data 1 & 4.13 & 0.9 & 5.84 & 3.64 & 0.9 & \textbf{4.34} & 0.62 & 3.77 & 0.9 & 4.56 & 0.62 & \textbf{3.13} & 0.86 & 4.38 & \textbf{0.7} \\
Data 2 & \textbf{2.53} & 0.89 & 5.27 & 2.83 & 0.89 & \textbf{3.69} & \textbf{0.76} & 2.9 & 0.88 & 4.44 & 0.74 & 3.72 & 0.89 & 4.88 & \textbf{0.76} \\
bike & 1.23 & 0.91 & 1.81 & 1.19 & 0.94 & \textbf{1.43} & 0.54 & \textbf{1.08} & 0.91 & 1.64 & \textbf{0.67} & 2.1 & 0.89 & 2.56 & 0.26 \\
bio & 1.52 & 0.89 & 2.14 & \textbf{1.51} & 0.9 & \textbf{2.11} & \textbf{0.5} & 1.54 & 0.9 & 2.14 & 0.46 & 2.15 & 0.89 & 2.4 & 0.0 \\
community & 1.86 & 0.89 & 2.88 & 1.69 & 0.92 & \textbf{2.41} & 0.63 & \textbf{1.59} & 0.9 & 2.43 & \textbf{0.64} & 2.0 & 0.89 & 2.5 & 0.51 \\
concrete & 0.62 & 0.9 & 0.92 & 0.63 & 0.92 & \textbf{0.81} & 0.56 & \textbf{0.56} & 0.87 & 0.89 & \textbf{0.71} & 1.17 & 0.87 & 1.43 & 0.0 \\
homes & 0.62 & 0.91 & 1.15 & 0.66 & 0.93 & \textbf{0.94} & 0.64 & \textbf{0.56} & 0.91 & 1.02 & \textbf{0.68} & 1.19 & 0.88 & 1.93 & 0.26 \\
star & \textbf{0.17} & 0.89 & \textbf{0.21} & \textbf{0.17} & 0.9 & \textbf{0.21} & 0.36 & \textbf{0.17} & 0.89 & \textbf{0.21} & \textbf{0.51} & 0.19 & 0.91 & 0.22 & 0.18 \\
\hline
\end{tabular}
}
\caption{Empirical results for forest-based methods on real and synthetic data as the average over five trials on width, coverage, ISL, and proportion better (P.B.) for $\alpha=0.1$. Calibration sets all have the size $n=500$.} \label{table:all-tuned-results-forest-small-tree}
\end{table}

\paragraph{Comparison with similar methods}
We compare our approaches with other methods such as \citet{izbicki2022cd}, \citet{cabezas2025regression}, \citet{martinez2024identifying}. These methods are similar to our approach as they also find data-driven groups on additional datasets for group-conditional conformal prediction. We denote by LOCART and LOFOREST the methods in \citet{cabezas2025regression}, CD-split in \citet{izbicki2022cd}, and TrustyAI Multi Region in \citet{martinez2024identifying}. We use the implementations publically provided by the authors \citet{izbicki2022cd}, \citet{cabezas2025regression}, and \citet{martinez2024identifying}. Note that these methods include their own tuning mechanism/hyperparameter optimization to control the tree structure. For example, LOCART and LOFOREST use cross-validation to tune the complexity parameter used for minimal cost-complexity pruning to minimize the mean-square error in predicting conformity scores. The multi-region method similarly optimizes tree hyperparameters (pruning complexity parameter, minimum number of samples per leaf, etc) via K-fold cross validation to optimize the coverage ratio. The results are in Table \ref{table:all-tuned-results_competitor}, which shows that our approach outperforms despite the conceptual similarity. We think that our improvement is attributed to the robust tree, which frees us from any need to split the calibration dataset to find the group information.

\begin{table}[!t]
\centering
\resizebox{\columnwidth}{!}{
\begin{tabular}{|l|ccc|cccc|cccc|cccc|cccc|cccc|cccc|}
\hline
 & \multicolumn{3}{c|}{Split Conformal} & \multicolumn{4}{c|}{Conformal Tree} & \multicolumn{4}{c|}{Conformal Forest} & \multicolumn{4}{c|}{LOCART} & \multicolumn{4}{c|}{LOFOREST} & \multicolumn{4}{c|}{CD-split} & \multicolumn{4}{c|}{TrustyAI Multi Region} \\
 & Width & Cov. & ISL & Width & Cov. & ISL & P.B. & Width & Cov. & ISL & P.B. & Width & Cov. & ISL & P.B. & Width & Cov. & ISL & P.B. & Width & Cov. & ISL & P.B. & Width & Cov. & ISL & P.B. \\
\hline
Data 1 & 4.16 & 0.88 & 6.12 & \textbf{3.27} & 0.87 & 4.26 & 0.6 & 3.52 & 0.89 & \textbf{4.11} & 0.66 & 3.99 & 0.88 & 6.19 & 0.8 & 3.99 & 0.88 & 6.19 & 0.8 & 4.12 & 0.82 & 11.1 & 0.69 & 3.78 & 0.85 & 6.15 & \textbf{1.0} \\
Data 2 & 2.85 & 0.9 & 7.02 & 3.02 & 0.86 & 5.52 & 0.68 & 3.87 & 0.91 & \textbf{4.55} & 0.75 & \textbf{2.64} & 0.88 & 7.06 & \textbf{0.8} & \textbf{2.64} & 0.88 & 7.06 & \textbf{0.8} & 4.5 & 0.89 & 12.36 & 0.47 & 2.72 & 0.88 & 7.08 & 0.6 \\
bike & 1.3 & 0.9 & 1.98 & 1.34 & 0.88 & 2.02 & 0.57 & \textbf{1.11} & 0.92 & \textbf{1.31} & \textbf{0.76} & 1.16 & 0.9 & 1.48 & 0.67 & 1.18 & 0.91 & 1.49 & 0.64 & 2.25 & 0.86 & 4.33 & 0.08 & 1.13 & 0.88 & 1.58 & 0.64 \\
bio & 1.66 & 0.9 & 2.21 & 1.67 & 0.9 & 2.25 & 0.73 & \textbf{1.51} & 0.92 & \textbf{1.93} & 0.56 & 1.56 & 0.9 & 2.1 & 0.52 & 1.57 & 0.91 & 2.07 & 0.53 & 4.07 & 0.61 & 0.78 & \textbf{0.78} & 27.61 & 1.0 & 27.61 & 0.0 \\
community & 1.95 & 0.91 & 2.81 & 1.66 & 0.87 & 2.6 & 0.48 & \textbf{1.65} & 0.9 & \textbf{2.32} & \textbf{0.67} & 1.74 & 0.91 & 2.45 & \textbf{0.67} & 1.76 & 0.93 & 2.43 & 0.6 & 2.08 & 0.85 & 4.85 & 0.47 & 1.99 & 0.94 & 2.5 & 0.56 \\
concrete & \textbf{0.61} & 0.91 & 0.84 & 0.71 & 0.89 & 0.92 & \textbf{0.7} & 0.77 & 0.93 & 0.86 & 0.54 & 0.62 & 0.92 & \textbf{0.82} & 0.38 & \textbf{0.61} & 0.91 & 0.84 & 0.4 & 1.9 & 0.95 & 2.18 & 0.0 & 1.28 & 0.99 & 1.31 & 0.0 \\
homes & \textbf{0.61} & 0.9 & 1.21 & 0.63 & 0.89 & 1.03 & 0.57 & 0.62 & 0.92 & \textbf{0.82} & \textbf{0.72} & 0.64 & 0.91 & 0.93 & 0.69 & 0.63 & 0.92 & 0.91 & 0.68 & 1.26 & 0.84 & 3.42 & 0.0 & 72.94 & 1.0 & 72.94 & 0.0 \\
star & \textbf{0.17} & 0.91 & \textbf{0.21} & \textbf{0.17} & 0.89 & 0.22 & \textbf{0.5} & 0.18 & 0.9 & \textbf{0.21} & 0.46 & \textbf{0.17} & 0.91 & \textbf{0.21} & 0.47 & \textbf{0.17} & 0.91 & \textbf{0.21} & 0.21 & \textbf{0.17} & 0.85 & 0.24 & 0.39 & 23.02 & 1.0 & 23.02 & 0.0 \\
\hline
\end{tabular}
}
\caption{Empirical results on real and synthetic data as the average over five trials on width, coverage, and proportion better (P.B.) for $\alpha=0.1$.} \label{table:all-tuned-results_competitor}
\end{table}

\paragraph{Comparison of prediction set sizes for skin disease prediction}
For the skin disease application in Section \ref{sec:skin_ex}, we also draw the frequency of the prediction set size in Figure \ref{fig:combined_derm_additional}. We can see that the frequency of a single element set is the highest for conformal tree and more concentrated near smaller set sizes, whereas most of predictions of Naive UQ predominate size 5.   

\begin{figure}[!h]
    \centering
    \includegraphics[width=0.5\linewidth]{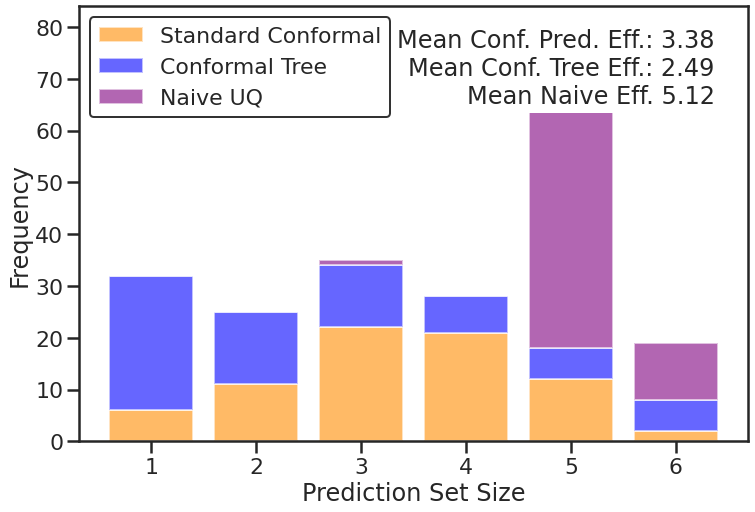}
    \caption{Comparison of prediction set sizes for erythemato squamous disease using GPT-4o. Conformal Tree yields prediction sets that are smaller than standard conformal prediction, containing 1.33 less items on average. }
    \label{fig:combined_derm_additional}
\end{figure}

\paragraph{Interpretable box-wise visualization}
For the smaller tree $(m=50)$ in Section \ref{sec:skin_ex}, we visualize its empirical test coverage, threshold, and prediction confidence in Figure \ref{fig:boxwise_vis}. Note that on Box 3, the model prediction quality is high (Figure \ref{fig:boxwise_vis} (i)). On this box, the threshold values (Orange line in Figure \ref{fig:boxwise_vis} (f)) are relatively very small compared to those of split conformal prediction.

\begin{figure}
    \centering
    \includegraphics[width=\linewidth]{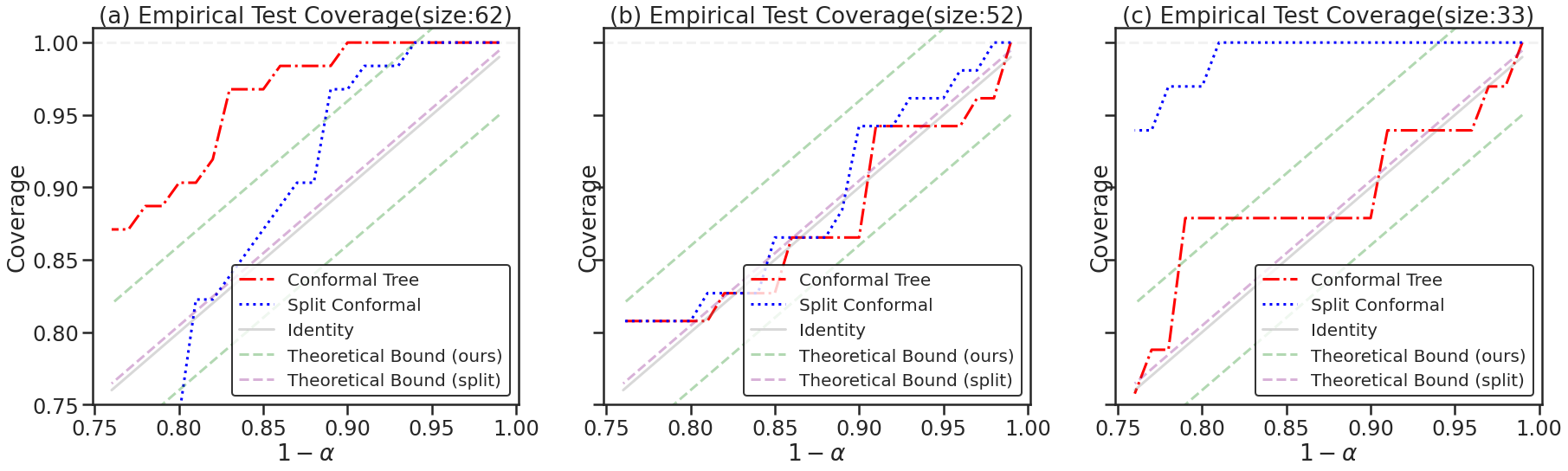}
    \includegraphics[width=\linewidth]{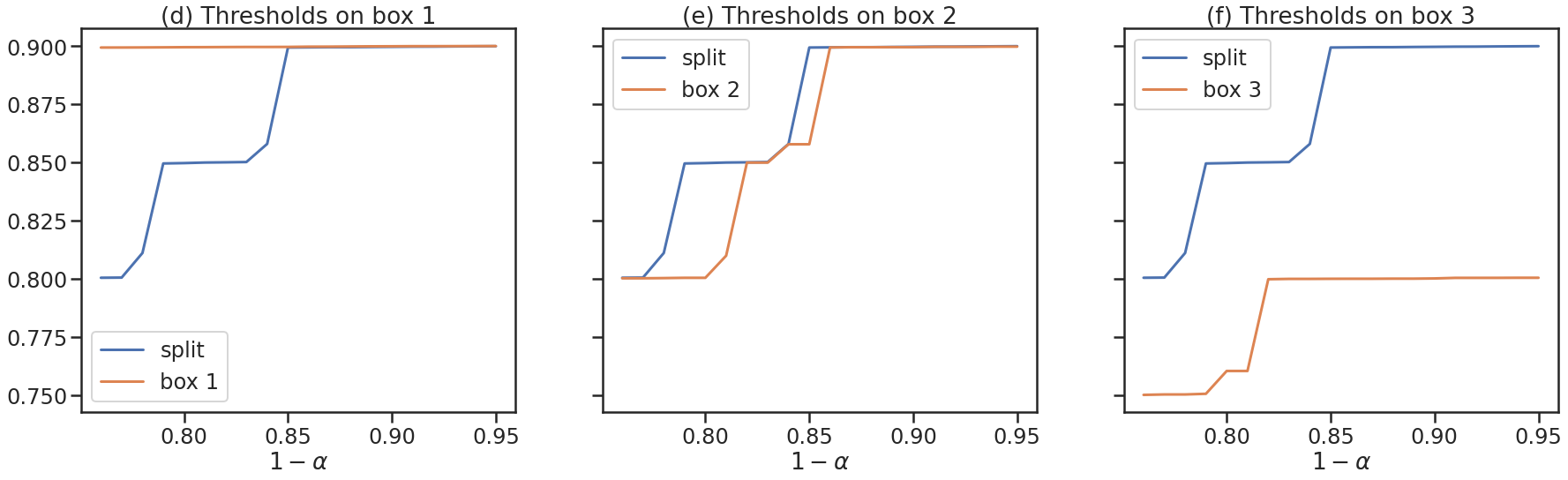}
    \includegraphics[width=\linewidth]{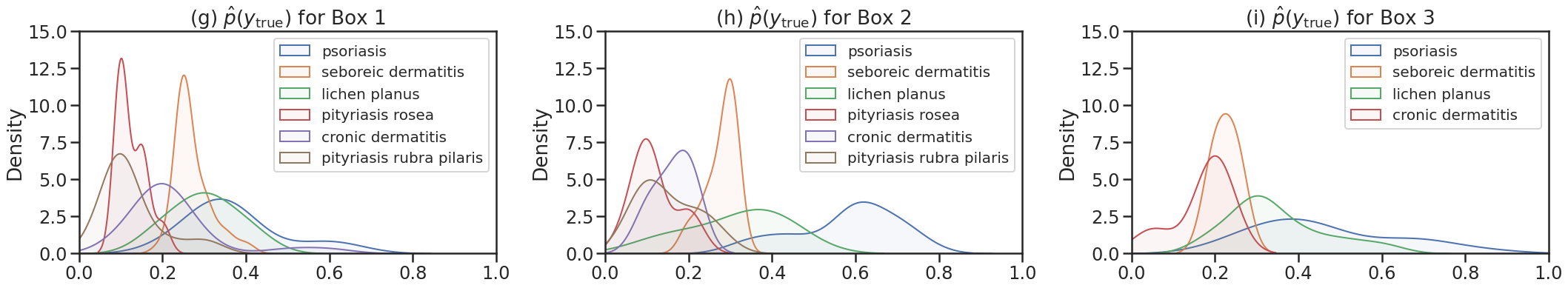}
    \caption{Empirical test coverage, threshold, and prediction confidence on each box were obtained by applying Conformal Tree on one realization of the calibration dataset ($60\%$ of the data). The behavior of Conformal Tree in this example is very interpretable: In Box 3, our method has smaller thresholds, improving the coverage rate significantly. In Box 1, the behavior is the opposite. The empirical coverage estimates in (a) are crude due to the lack of samples (see, Figure \ref{fig:emp-coverage-test}). For the splitting process of this tree, see Figure \ref{fig:derm-tree-structure2}.}
    \label{fig:boxwise_vis}
\end{figure}

\paragraph{Empirical test coverage}
In Figure \ref{fig:emp-coverage-test}, we plot the empirical test coverage of split conformal and Conformal Tree on the GPT diagnosis example in Section \ref{sec:skin_ex}. For a single split (Figure \ref{fig:emp-coverage-test} Left), the empirical estimation of the coverages even for the split conformal prediction has a high variation. We think it was because of the number of test data points is small. To alleviate it, we repeat random splits and apply these conformal methods and take average over the estimated coverage for each $\alpha$ value. In Figure \ref{fig:emp-coverage-test} Right shows, as the number of random splits increases, the coverage aligns more with the identity line. Note that in any ways, Conformal tree always empirically covers at the desired theoretical coverage level.

\begin{figure}[h!]
    \centering
    \includegraphics[width=0.4\linewidth]{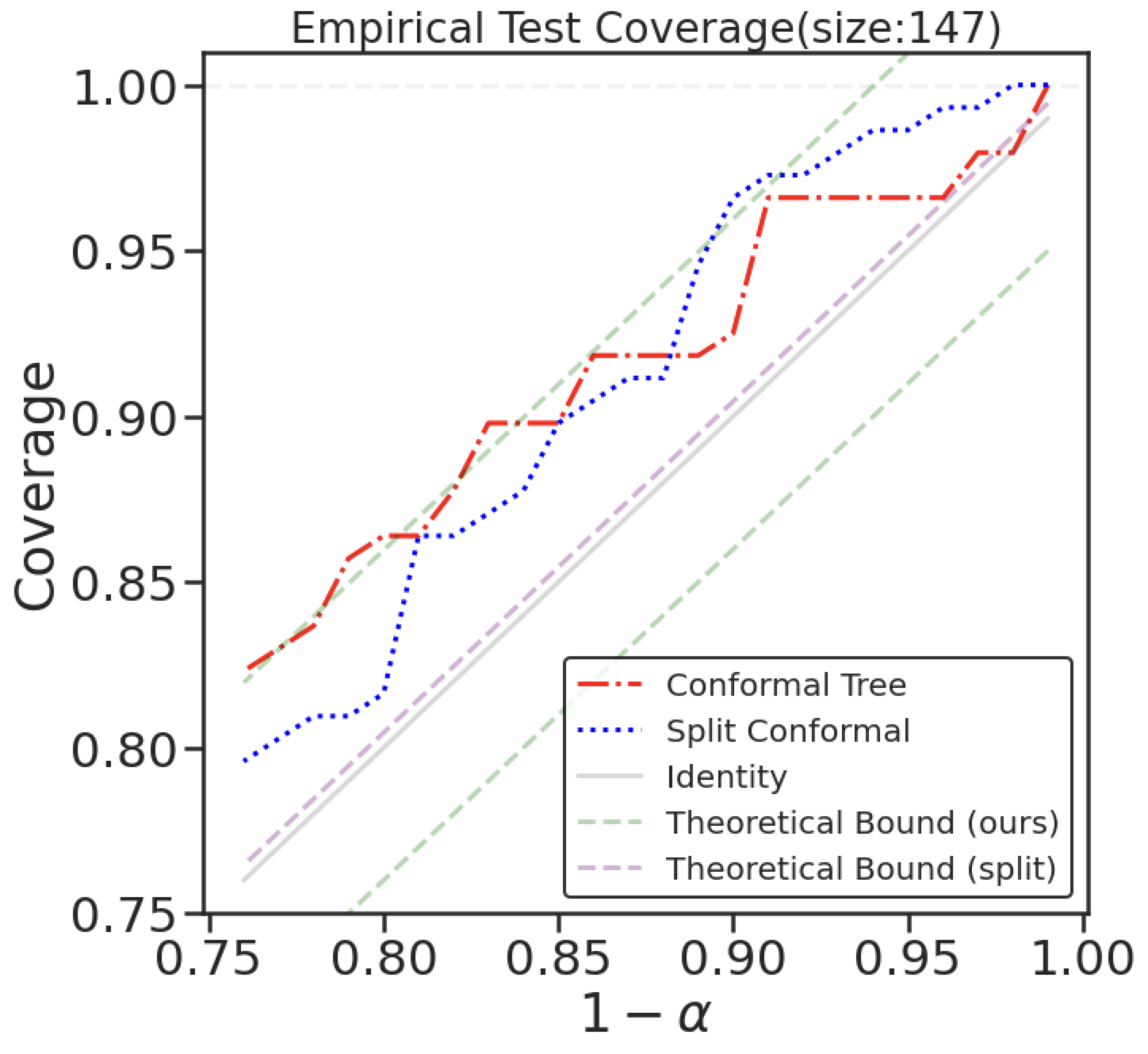}
    \includegraphics[width=0.4\linewidth]{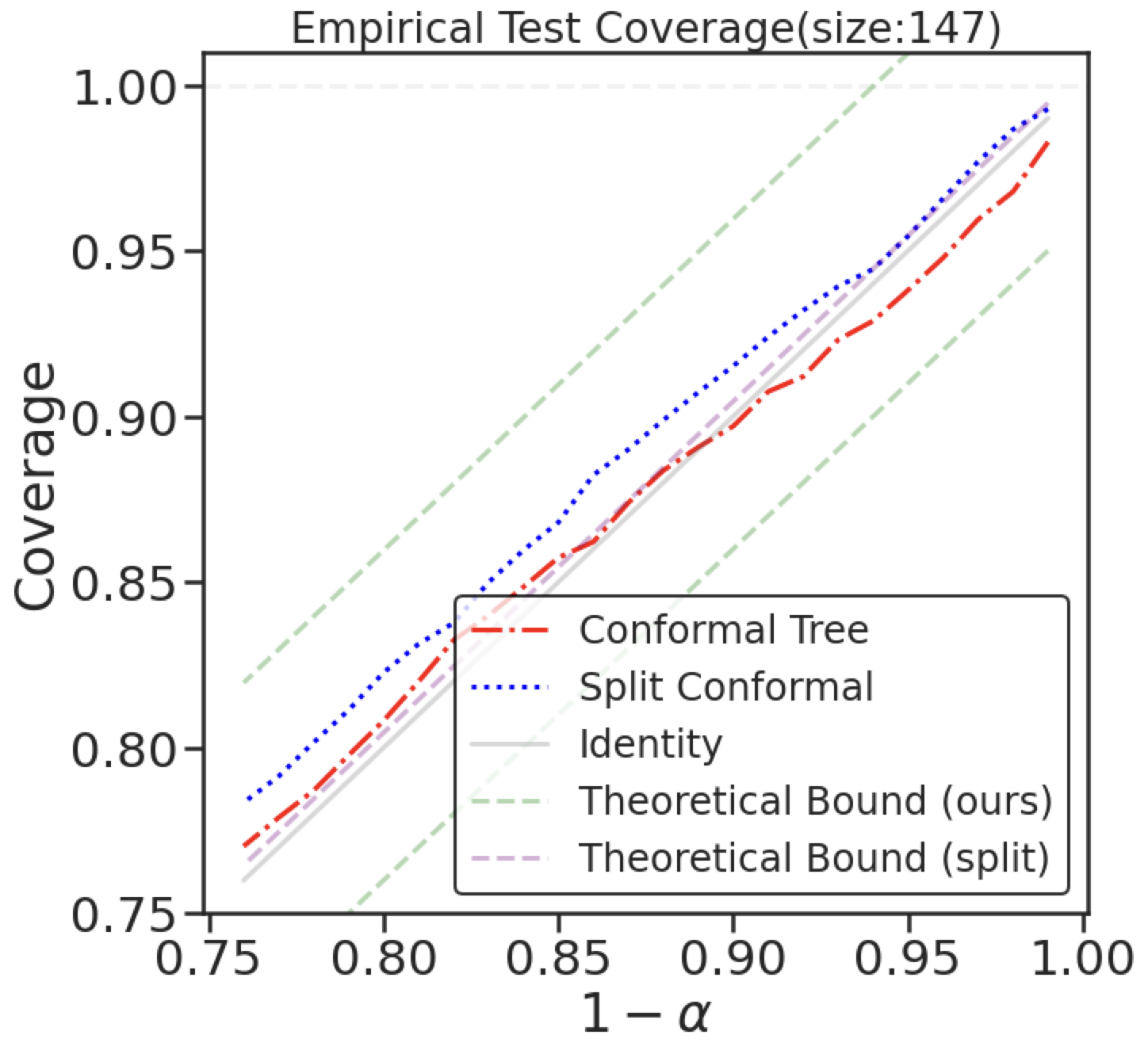}
    \caption{Empirical test coverage for split-conformal and Conformal Tree for varying $1-\alpha$ values on the GPT diagnosis example. Left: The empirical coverage rate for a single split. Right: The empirical (average) coverage rate for 40 random splits.}
    \label{fig:emp-coverage-test}
\end{figure}

\paragraph{Conformal Tree's tree structure}

In Figure \ref{fig:derm-tree-structure}, we additionally provide an interpretability visualization of the tree that corresponds to \ref{fig:combined_derm}. It shows each level of Conformal Tree's structure on the skin disease calibration data similar to Figure \ref{fig:derm-tree-structure2}.

\begin{figure}[!h]
    \centering
    \includegraphics[width=\linewidth]{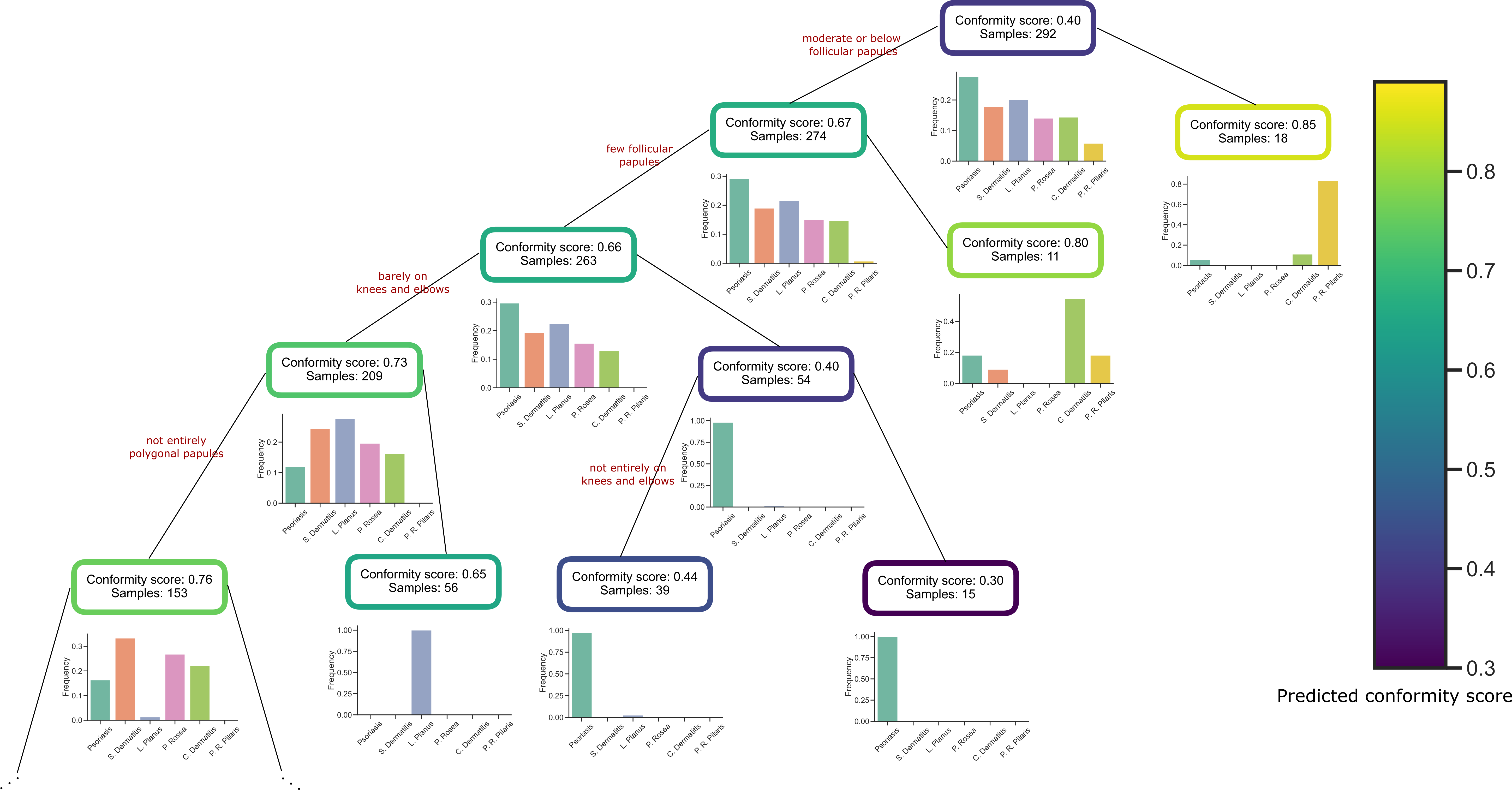}
    \caption{First five levels of Conformal Tree's structure on the skin disease calibration data, with $m=10$ and $K_{max}=20$. Node colors correspond to the mid-range of conformity scores within the node. Below each node, we plot the empirical distribution of calibration data corresponding to each disease within that node. 
    }
    \label{fig:derm-tree-structure}
\end{figure}

\end{appendices}

\end{document}